\documentclass[12pt,a4paper,pdftex,english]{scrartcl}
\pdfoutput=1
\usepackage[latin1]{inputenc}
\usepackage[english]{babel}	
\usepackage{amsmath}
\usepackage{amsfonts}
\usepackage{amssymb}
\usepackage{amsthm}
\usepackage{algorithm}
\usepackage{algorithmic}
\usepackage{enumitem}
\usepackage{csquotes}
\usepackage{color}
\usepackage{graphics}
\usepackage{url}
\usepackage{tikz,tabularx}
\usetikzlibrary{arrows,decorations.pathreplacing,backgrounds,positioning,fit,shapes,calc}

\pgfkeys{%
/tikz/background rectangle/.style= {rounded corners,draw,fill=white},
/tikz/vertex/.style={circle,draw=black,fill=black, minimum size=1.5mm,inner sep=0pt},%
/tikz/vertex-out/.style={circle,draw=blue!50,fill=blue!5,minimum size=7mm},% 
/tikz/edge/.style={->,shorten <=1pt,>=stealth,thick},%
/tikz/undiredge/.style={-,>=stealth,thick},%
/tikz/edge-out/.style={->,shorten <=1pt,>=stealth,thick,densely dotted},%
/tikz/graphframe/.style={rectangle,draw,minimum height=5cm,minimum width=3cm},
/tikz/unlabvertex/.style={circle,draw=black,fill=black, minimum size=1.5mm,inner sep=0pt},%
/tikz/vector/.style={rectangle split,draw,rectangle split parts=3,rectangle split horizontal,rectangle split part fill={red!50, red!30, red!50},draw},%
/tikz/inner/.style={circle,draw=black,minimum size=7mm},
/tikz/sink/.style={rectangle,draw=black,minimum size=7mm},
/tikz/zero/.style={->,shorten <=1pt,>=stealth,semithick,densely dotted},
/tikz/one/.style={->,shorten <=1pt,>=stealth,semithick,solid},
/tikz/tnode/.style={circle,draw=black,inner sep=2pt},
/tikz/bits/.style={rectangle,draw=black,text height=10pt,text depth=3pt,inner sep=3pt,outer sep=0pt},
/tikz/entry/.style={rectangle,inner sep=2pt,minimum size=6mm}
}

\newcommand{\rect}[5]{    \draw[thick] ($ (#1.north west)+(#5pt,-#5pt) $) -- ($ (#5pt,#5pt)+(#2.south west) $);
    \draw[thick] ($ (#2.south west)+(#5pt,#5pt) $) -- ($ (#3.south east)+(-#5pt,#5pt) $);
    \draw[thick] ($ (#3.south east)+(-#5pt,#5pt) $) -- ($ (#4.north east)+(-#5pt,-#5pt) $);
    \draw[thick] ($ (#4.north east)+(-#5pt,-#5pt) $) -- ($ (#1.north west)+(#5pt,-#5pt) $);}
    
\newcommand{\rectd}[5]{    \draw[thick,dashed] ($ (#1.north west)+(#5pt,-#5pt) $) -- ($ (#5pt,#5pt)+(#2.south west) $);
    \draw[thick,dashed] ($ (#2.south west)+(#5pt,#5pt) $) -- ($ (#3.south east)+(-#5pt,#5pt) $);
    \draw[thick,dashed] ($ (#3.south east)+(-#5pt,#5pt) $) -- ($ (#4.north east)+(-#5pt,-#5pt) $);
    \draw[thick,dashed] ($ (#4.north east)+(-#5pt,-#5pt) $) -- ($ (#1.north west)+(#5pt,-#5pt) $);}

\newtheoremstyle{aufschrieb}   
  {0.5cm}                 %Space above    
  {0.5cm}                 %Space below    
  {\itshape}                         %Body font: original {\normalfont}    
  {}                         %Indent amount (empty = no indent,%\parindent = paraindent)    
  {\normalfont\bfseries}  %Thm head font original 
  {\normalfont :} 
  {\newline}
  {}
\theoremstyle{aufschrieb}

\definecolor{commentgreen}{rgb}{0.1,0.1,0.6}

\floatname{algorithm}{Algorithm}

\newtheorem{algo}{Algorithm}[section]

\newtheorem{theorem}[algo]{Theorem}

\newtheorem{definition}[algo]{Definition}

\newtheorem{observation}[algo]{Observation}

\newcommand{\bset}{\lbrace 0,1 \rbrace}

\DeclareMathOperator{\mod2}{mod}

\newcommand{\eg}{e.\,g., }
\newcommand{\Eg}{E.\,g., }

\newcommand{\ie}{i.\,e., }
\newcommand{\Ie}{I.\,e., }

\begin{document}

\title{OBDD-Based Representation of Interval Graphs}

\author{Marc Gill\'{e}\thanks{Supported by Deutsche Forschungsgemeinschaft, grant BO 2755/1-1.}\\{\normalsize TU Dortmund, LS2 Informatik, Germany}}

\maketitle

\begin{abstract}
\noindent A graph $G = (V,E)$ can be described by the characteristic function of the edge set $\chi_E$ which maps a pair of binary encoded nodes to $1$ iff the nodes are adjacent. Using \emph{Ordered Binary Decision Diagrams} (OBDDs) to store $\chi_E$ can lead to a (hopefully) compact representation. Given the OBDD as an input, symbolic/implicit OBDD-based graph algorithms can solve optimization problems by mainly using functional operations, \eg quantification or binary synthesis. While the OBDD representation size can not be small in general, it can be provable small for special graph classes and then also lead to fast algorithms. In this paper, we show that the OBDD size of unit interval graphs is $O(\vert V \vert /\log \vert V \vert)$ and the OBDD size of interval graphs is $O(\vert V \vert \log \vert V \vert)$ which both improve a known result from Nunkesser and Woelfel (2009). Furthermore, we can show that using our variable order and node labeling for interval graphs the worst-case OBDD size is $\Omega(\vert V \vert \log \vert V \vert)$. We use the structure of the adjacency matrices to prove these bounds. This method may be of independent interest and can be applied to other graph classes. We also develop a maximum matching algorithm on unit interval graphs using $O(\log \vert V \vert)$ operations and a coloring algorithm for unit and general intervals graphs using $O(\log^2 \vert V \vert)$ operations and evaluate the algorithms empirically.
\end{abstract}

\section{Introduction}
\label{sec:intro}
The development of graph algorithms is a classic and intensively studied area of computer science. But the requirements on graph algorithms have changed by the emergence of massive graphs, \eg the internet graph or social networks. The representation size of a graph with $N$ nodes and $M$ edges given as adjacency matrix or lists is $\Theta(N^2)$ or $\Theta(N + M)$. There are applications, \eg dealing with a state transition graphs in circuit verification, where even polynomial running time may not be feasible or the input does not fit into the main memory. In order to deal with such massive graphs, \emph{symbolic} or \emph{implicit} graph algorithms have been investigated, where the input is represented by the characteristic function of the edge set. Identifying the nodes by binary numbers, the characteristic function becomes a Boolean function, which can be represented by \emph{Ordered Binary Decision Diagrams} (OBDDs), which are a well known and commonly used data structure for Boolean functions. OBDDs were introduced by Bryant \cite{Bryant86} and support many important functional operations on Boolean functions efficiently. Therefore, a research area came up concerning the design and analysis of implicit/symbolic (graph) algorithms on OBDD represented inputs (\cite{CLMD94, HachtelS97, GentiliniPP03, Sawitzki04, SawitzkiSOFSEM06, SawitzkiLATIN06, Woe2006, GentiliniPP08}). A motivation of this line of research is that implicit representations can be significantly smaller than explicit representations on structured graphs, thus enabling the algorithms to process larger amounts of data. In particular, processing data in a more compact form might speed up the time needed for the algorithms.

In theory, problems on implicitly represented inputs become harder than their explicit equivalent. Even the $s$-$t$-connectivity problem, \ie the decision whether two nodes $s$ and $t$ of an undirected graph are connected, is PSPACE-complete on OBDD-represented graphs \cite{FeigenbaumKVV99} while the explicit variant is in $L$, the complexity class consisting of all problems decidable by a logspace Turing machine. Nevertheless, implicit algorithms are successful in many practical applications, \eg model checking \cite{BurchCMDH92}, integer linear programming \cite{LaiPV94} and logic minimization \cite{Coudert95}, and can be seen as a kind of heuristic (regarding time and space) to compute an optimal solution for problems on very large instances. One of the first implicit algorithms for a classical graph problem was the maximum flow algorithm on $0$-$1$-networks presented in \cite{HachtelS97}. There, Hachtel and Somenzi were able to solve instances up to $10^{36}$ edges and $10^{27}$ nodes in reasonable time. Sawitzki \cite{Sawitzki04} described another implicit algorithm for the same problem, which uses $O(N \log^2 N)$ functional operations. 

The number of operations used in an implicit algorithm is an important measure of difficulty \cite{BloemGS06} but it is also known \cite{BloemGS06,HojatiTKB93} that an implicit algorithm computing the transitive closure, that uses an iterative squaring approach and a polylogarithmic number of operations, is often inferior to an implicit sequential algorithm, which needs in worst case a linear number of operations. In this case the advantage of the small number of functional operations is canceled out by the probably large sizes of the intermediate OBDDs. In order to analyze the actual running time of an OBDD-based graph algorithm, it is crucial to determine both the number of functional operations and the sizes of the OBDDs which are generated during the computation.

Sawitzki \cite{SawitzkiSOFSEM06, Sawitzki07} showed that all problems which are decidable by polynomially many processors using polylogarithmic time (\ie which are in $NC$) are computable by an implicit algorithm using a polylogarithmic number of functional operations on a logarithmic number of Boolean variables. This is a quite structural result and does not lead to either an efficient transformation of parallel algorithms into implicit algorithms or a guarantee that implicit algorithms using a polylogarithmic number of functional operations perform well in practice. 

Implicit algorithms using a polylogarithmic number of operations were designed for instance for topological sorting \cite{Woe2006}, maximal matching \cite{BolligP12} and minimum spanning tree \cite{Bollig12} where a matching $M$, \ie a set of edges without a common vertex, is called maximal if $M$ is no proper subset of another matching.
However, non trivial bounds on the sizes of the OBDDs are hard to determine and, with it, the actual running time of an implicit algorithm. Only on very structured graphs like grid graphs good analysis of the running time are known, \eg for the maximum flow algorithm \cite{Sawitzki04}, topological sorting \cite{Woe2006} and maximal and maximum matching \cite{BolligP12, BolligGP12} (a matching is called maximum matching if there is no matching consisting of a larger number of edges). As a consequence, the practical performance of implicit algorithms is often evaluated experimentally, \eg for the maximum matching problem in bipartite graphs \cite{BolligGP12} or for the maximum flow problem \cite{HachtelS97, Sawitzki04}.

For a good running time the input size of an implicit algorithm, \ie the size of the OBDD representing the input graph, should be small. Nunkesser and Woelfel \cite{NuWo09} showed that the size of an OBDD representing an arbitrary graph is $O(N^2 / \log N)$, which is similar to the space needed for a representation by adjacency matrix. For bipartite graphs they were able to show a lower bound of $\Omega(N^2 / \log N)$, which means that there is a bipartite graph which OBDD size is bounded below by $\Omega(N^2 / \log N)$. Notice that a lower bound of the OBDD size of a graph class does not mean that the OBDD size of every graph from this class is bounded below by this value.

Nunkesser and Woelfel \cite{NuWo09} also investigated other restricted graph classes such as interval graphs. An interval graph is an intersection graph of intervals on the real line, \ie two intervals (nodes) are adjacent iff they have a nonempty intersection. If these intervals have a length of $1$, then the graph is called unit interval graph. (Unit) Interval graphs were extensively studied and have many applications, \eg in genetics, archaeology, scheduling, and much more \cite{Golumbic04}. Nunkesser and Woelfel \cite{NuWo09} proved that general interval graphs can be represented by OBDDs of size $O(N^{3/2} \log^{3/4} N)$ and $O(N / \sqrt{\log N})$ for unit interval graphs. Due to counting arguments, they proved a lower bound of $\Omega(N)$ for general interval graphs and $\Omega(N/\log N)$ for unit interval graphs.

As in \cite{NuWo09}, we use $n = \lceil \log N \rceil$ bits, \ie the minimal number of bits, to encode the nodes of a graph. Since the worst-case OBDD size is exponentially large in the number of input bits, using $\chi_E$ in an implicit algorithm motivates to use a minimal amount of input bits to avoid a large worst-case OBDD size.
Using a larger domain for the labels also possibly increases the size of the data structure storing the valid labels which is often needed in implicit algorithms.
Aiming for a good compression of $\chi_E$, Meer and Rautenbach \cite{MeerR09} investigated graphs with bounded clique-width or tree-width and increases the number of bits used for the node labeling to $c \cdot \log N$ with constant $c$ and were able to improve for instance the OBDD size of cographs from $O(N \log N)$ \cite{NuWo09} to $O(N)$.

We investigate implicit algorithms for coloring interval graphs, \ie coloring the nodes of an interval graph such that all adjacent nodes have different colors and the number of used colors is minimal, and for maximum matching on unit interval graphs. Coloring of interval graphs has applications in VLSI design and scheduling \cite{GuptaLL79} and there is an optimal greedy coloring algorithm \cite{Olariu91} which can be implemented in linear time. The first parallel matching algorithm for (general) interval graphs was given by a parallel algorithm for two processor scheduling \cite{HelmboldM87} using $N^{10}$ processors and $O(\log^2 N)$ time. It was improved step-by-step until in \cite{ChungPC97} the current best parallel algorithm was presented using $N^3/\log^4 N$ and $O(\log^2 N)$ time. Furthermore, they showed that it is possible to compute a maximum matching on unit interval graphs in parallel by $O(N / \log N)$ processors using $O(\log N)$ time.

\subsection{Our Contribution}
In Section 3 we present a new method to show upper and lower bounds of the size of an OBDD representing a graph. We sort the rows and columns of the adjacency matrix of a graph in such a way that each node of the OBDD (labeled with the same input variable) corresponds to a distinct block of this adjacency matrix. Thus, by counting the number of different blocks, which is probably easier than counting different subfunctions, we can bound the number of nodes of the OBDD. 

Using this method and some known structure of the adjacency matrix of interval graphs \cite{Mertzios08}, we improve the upper bound on general interval graphs to \linebreak $O(N \log N)$ while using a more convenient way to label the nodes than in \cite{NuWo09}. Using a probabilistic argument, we prove that the worst-case OBDD size is $\Omega(N \log N)$ if we use the same labeling and variable order as in our upper bound. We can close the gap of the upper bound and the lower bound in the case of unit interval graphs and show that the OBDD size is $O(N /\log N)$.

In Section 4 we present two implicit algorithms for (unit) interval graphs: A maximum matching algorithm for unit interval graphs using only $O(\log N)$ functional operations and a coloring algorithm for interval graphs using $O(\log^2 N)$ functional operations.
The matching algorithm takes advantage of the information given by the labels of the nodes. Furthermore, we were able to compute the transitive closure of a unit interval graph using only $O(\log N)$ operations instead $O(\log^2 N)$ operations, which are needed in general. In order to implement this algorithm efficiently, we have to extend a known result due to Woelfel \cite{Woe2006} to a different variable order for constructing OBDDs representing multivariate threshold functions.
For the coloring algorithm we show how to get a total order on the right endpoints (given that the labels of the nodes respect the order of the left endpoints) and how to compute a minimal coloring of the nodes by using these orders based on a optimal greedy algorithm \cite{Olariu91}.
To the best of the author's knowledge, this is the first time that the labeling of nodes is used to speed up an implicit algorithm for a large graph as interval graphs and to improve the number of functional operations. In Section 5 we evaluate the implicit algorithms experimentally and see that the matching algorithm is both very fast and space efficient while, unfortunately, the coloring algorithm does not perform well on both unit and general interval graphs. The poor performance of the coloring algorithm is very likely due to the fact that the implicit algorithm is simulating the sequential coloring algorithm. Nevertheless, it uses some nice ideas to accomplish this which differ from the parallel implementation ideas.

A simple implicit representation of an interval graph with $N$ nodes is a list of $N$ intervals and needs $\Theta(N)$ space. Our result on the OBDD size of interval graphs shows that in the worst case the OBDD representation is almost as good as the interval representation with the advantage that it is possible to use $o(N)$ space for some instances. Together with the experiments, this shows that the representation of at least unit interval graphs with OBDDs enables a good compression without loosing the usability in algorithms.

\section{Preliminaries}
\label{sec:pre}
\subsection{OBDDs}
\label{sec:pre_obdd}
We denote the set of Boolean functions $f: \bset^n \rightarrow \bset$ by $B_n$. Let $(x_0,\ldots,x_{n-1}) = x \in \bset^n$ be a binary number of length $n$ and $\vert x \vert := \sum_{i=0}^{n-1} x_i \cdot 2^i$ the value of $x$. Further, let $l \in \mathbb{N}$ be a natural number then we denote by $[l]_2$ the corresponding binary number of $l$, \ie $\vert [l]_2 \vert = l$.

Let $G = (V,E)$ be a directed graph with node set $V = \lbrace v_0, \ldots, v_{N-1} \rbrace$ and edge set $E \subseteq V \times V$. Here, an undirected graph is interpreted as a directed symmetric graph. Implicit algorithms are working on the characteristic function $\chi_E \in B_{2n}$ of $E$ where $n = \lceil \log N \rceil$ is the number of bits needed to encode a node of $V$ and $\chi_E(x,y) = 1$ if and only if $(v_{\vert x \vert}, v_{\vert y \vert}) \in E$. In order to deal with Boolean functions, OBDDs were introduced by Bryant \cite{Bryant86} to get a compact representation (see Fig.\,\ref{fig:obdd}), which supports a bunch of functional operations efficiently.
\begin{definition}[Ordered Binary Decision Diagram (OBDD)]

\textbf{Order.} A \emph{variable order} $\pi$ on the input variables $X = \lbrace x_0, \ldots, x_{n-1} \rbrace$ of a Boolean function $f \in B_n$ is a permutation of the index set $I = \lbrace 0, \ldots, n-1 \rbrace$.

\textbf{Representation.} A \emph{$\pi$-OBDD} is a directed, acyclic and rooted graph $G$ with two sinks labeled by the constants $0$ and $1$. Each inner node is labeled by an input variable from $X$ and has exactly two outgoing edges labeled by $0$ and $1$. Each edge $(x_i,x_j)$ has to respect the variable order $\pi$, \ie $\pi(i) < \pi(j)$.

\textbf{Evaluation.} An assignment $a \in \bset^n$ of the variables defines a path from the root to a sink by leaving each $x_i$-node via the $a_i$-edge. A $\pi$-OBDD $G_f$ represents $f$ iff for every $a \in \bset^n$ the defined path ends in a sink with label \nolinebreak $f(a)$.

\textbf{Complexity.} The \emph{size} of a $\pi$-OBDD $G$, denoted by $\vert G \vert$, is the number of nodes in $G$. The $\pi$-OBDD size of a function $f$ is the minimum size of a $\pi$-OBDD representing $f$. The OBDD size of $f$ is the minimum $\pi$-OBDD size over all variable orders $\pi$. The \emph{width} of $G$ is the maximum number of nodes labeled by the same input variable.
\end{definition}
\begin{figure}
\begin{center}
\begin{tikzpicture}
	\node (x1)   at (2,4)   [inner] {$x_1$};
	\node (y1-1) at (1,3)   [inner] {$y_1$}; 
	\node (y1-2) at (3,3)   [inner] {$y_1$}; 
	\node (x0)   at (1.7,2)   [inner] {$x_0$}; 

    \node (y0-2) at (2.5,1)   [inner] {$y_0$}; 
	\node (0)    at (1,0)   [sink]  {$0$};
	\node (1)    at (3,0)   [sink]  {$1$};

	\draw [zero] (x1) to (y1-1);
	\draw [one]  (x1) to (y1-2);
	\draw [zero] (y1-1) to (x0);
	\draw [one]  (y1-1) to (0);
	\draw [zero] (y1-2) to (1);
	\draw [one]  (y1-2) to (x0);
	
	\draw [zero] (x0) to (0);
	\draw [one]  (x0) to (y0-2);
	\draw [zero] (y0-2) to (1);
	\draw [one]  (y0-2) to (0);	
	
\end{tikzpicture}
\end{center}
\caption{A minimal $\pi$-OBDD representing the function $GT(x,y)$ with $GT(x,y) = 1$ iff $\vert x \vert > \vert y \vert$ and $\pi = (x_1,y_1,x_0,y_0)$}
\label{fig:obdd}
\end{figure}
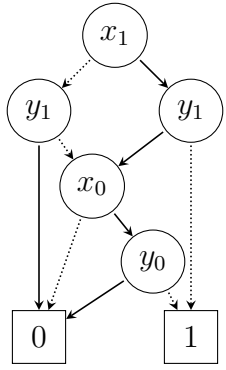
\noindent
In the following we describe a list of important operations on Boolean functions which we will use in this paper and give the time requirements in the case of OBDDs (see, \eg Section 3.3 in \cite{Wegener00} for a detailed list). 
Let $f$ and $g$ be Boolean functions in $B_n$ on the variable set 
$X=\{x_0, \ldots, x_{n-1}\}$ and let $G_f$ and $G_g$ be OBDDs representing $f$ and $g$, respectively.

\begin{enumerate}
\item \textbf{Negation:} Given $G_f$, compute an OBDD representing the function $\overline{f} \in B_n$. Time: $O(1)$
\item \textbf{Replacement by constant:} Given $G_f$, an index $i\in \{0, \ldots, n-1\}$, and a Boolean constant $c_i\in \{0,1\}$, compute an OBDD representing the subfunction $f_{|x_i=c_i}$. Time: $O(\vert G_f \vert)$
\item \textbf{Synthesis:} Given $G_f$ and $G_g$ and a binary Boolean operation $\otimes \in B_2$, compute an OBDD representing the function $h\in B_n$ defined as $h:=f\otimes g$. Time: $O(\vert G_f \vert \cdot \vert G_g \vert)$
\item \textbf{Quantification:} Given $G_f$, an index $i\in \{1, \ldots, n\}$ and a quantifier $Q\in\{\exists, \forall\}$, compute an OBDD representing the function $h\in B_n$ defined as $h:=Q x_i: f$ where  $\exists x_i: f := f_{|x_i=0}\vee f_{|x_i=1}$ and $\forall x_i: f := f_{|x_i=0}\wedge f_{|x_i=1}$. Time: see replacement by constant and synthesis
\end{enumerate}
\noindent
In the rest of the paper quantifications over $k$ Boolean variables 
$Q x_1, \ldots, x_{k}: f$ are denoted by $Q x: f$, where $x= (x_1, \ldots, x_k)$. The following operation (see, \eg \cite{SawitzkiLATIN06}) is 
useful to reverse the edges of a given graph, \ie let $\chi_E(x,y)$ be a directed graph and we want to compute $\chi_E(y,x)$ which represents the edge set $\lbrace (v_{\vert y \vert}, v_{\vert x \vert}) \mid (v_{\vert x \vert},v_{\vert y \vert}) \in E \rbrace$ consisting of the reverse edges of $E$.
\begin{definition}
Let $k \in \mathbb{N}$, $\rho$ be a permutation of $\{1,\ldots, k\}$ and
$f\in B_{kn}$ be defined on Boolean variable vectors $x^{(1)}, \ldots, x^{(k)}$ of length $n$. The argument reordering $\mathcal{R}_\rho(f)\in B_{kn}$ with respect to $\rho$ is defined by $\mathcal{R}_\rho(f)(x^{(1)}, \ldots, x^{(k)}) := f(x^{(\rho (1))}, \ldots, x^{(\rho (k))})$.
\end{definition}
\noindent
This operation can be computed by just renaming the variables and repairing the variable order using $3(k-1)n$ functional operations (see \cite{BolligLW96}). 

An important variable order is the interleaved variable order which is defined on vectors of length $n$ where the variables with the same significance are tested one after another.

\begin{definition}
Let $x^{(1)},\ldots,x^{(k)} \in \bset^n$ be $k$ input variable vectors of length $n$. Let $\pi$ be a permutation of $\lbrace 0,\ldots,n-1 \rbrace$. Then 
$$ \pi_{k,n} = (x_{\pi(0)}^{(1)},x_{\pi(0)}^{(2)},\ldots,x_{\pi(0)}^{(k)}, \ldots, x_{\pi(n-1)}^{(1)},\ldots,x_{\pi(n-1)}^{(k)}) $$
is called \emph{$k$-interleaved} variable order for $x^{(1)},\ldots,x^{(k)}$. If $\pi = (n-1,\ldots,0)$ then we say that the variables are tested with decreasing significance.
\end{definition}
\noindent
Choosing an interleaved variable order in OBDD-based algorithms is common practice, since auxiliary functions, \eg the equality or the greater than function, and multivariate threshold functions, which we will define in section \ref{sec:algig_thresh}, have to use an interleaved variable order for a compact OBDD representation. 

The input of an OBDD-based graph algorithm consists of the characteristic function $\chi_E$ represented by an OBDD and the output is a characteristic function $\chi_O$ of a set $O$ which is computed by mainly using functional operations. As we can see in the above listing of the operations, the running time of such algorithms depends on the actual size of the OBDDs which are used for an functional operation during the computation. In general, it is difficult to prove a good upper bound on the running time because we have to know a good upper bound on the size of every OBDD used as an input for an operation.

However, if the size of the OBDD representing the input graph is large, any implicit algorithm using this OBDD is likely to have an inadequate running time. Beside the variable order, the labeling of the nodes is another optimization parameter with huge influence on the input size. For OBDDs representing state transitions in finite state machines, Meinel and Theobald \cite{MeinelT99} showed that there can be an exponential blowup of the OBDD size from a good labeling to a worst-case labeling. Nevertheless, a small input OBDD size, \ie a good labeling of the nodes for some variable order, does not guarantee a good performance of the implicit algorithm since the sizes of the intermediate OBDDs do not have to be small. Indeed, an exponential blowup from input to output size is possible \cite{SawitzkiLATIN06,Bollig12}.

We denote by $f_{\mid x_{\pi(0)}=a_{\pi(0)},\ldots,x_{\pi(i-1)}=a_{\pi(i-1)}}$ the subfunction where $x_{{\pi(j)}}$ is replaced by the constant $a_{\pi(j)}$ for $0 \leq j \leq i-1$. The function $f$ \emph{depends essentially} on a variable $x_i$ iff $f_{\mid x_i = 0} \neq f_{\mid x_i = 1}$. A characterization of minimal $\pi$-OBDDs due to Sieling and Wegener \cite{SielingW93} can be often used to bound the OBDD size. 
\begin{theorem}[\cite{SielingW93}]
\label{thm:minimal_obdd}
Let $f \in B_n$ and for all $i=0,\ldots,n-1$ let $s_i$ be the number of different subfunctions which result from replacing all variables $x_{{\pi(j)}}$ with $0 \leq j \leq i-1$ by constants and which essentially depend on $x_{\pi(i)}$. Then the minimal $\pi$-OBDD representing $f$ has $s_i$ nodes labeled by $x_{\pi(i)}$, \ie the minimal $\pi$-OBDD has size $\sum_{i=0}^{n-1} s_i$.
\end{theorem}

\subsection{Basic Functions and Implicit Algorithms}
\label{sec:pre_basicalg}
It is well known that the OBDD size of the equality $EQ(x,y)$ and greater than function $GT(x,y)$ with $EQ(x,y) = 1 \Leftrightarrow \vert x \vert = \vert y \vert$ and $GT(x,y) = 1 \Leftrightarrow \vert x \vert > \vert y \vert$ is linear in the number of input bits for an interleaved variable order (see, \eg \cite{Wegener00}). It is also possible to construct the representing OBDDs for these functions in linear time. For the sake of code readability, we use $\vert x \vert = \vert y \vert$ and $\vert x \vert > \vert y \vert$ to denote $EQ(x,y)$ and $GT(x,y)$ in our algorithms. Furthermore, by $\vert x \vert > c$ ($\vert x \vert = c$) for some constant $c$ we denote the function $GT(x,y)_{\mid y = [c]_2}$ ($EQ(x,y)_{\mid y = [c]_2}$) where the $y$-variables are replaced by constants corresponding to the binary number $[c]_2$ of $c$.

Let $R(x,y) \in B_{2n}$ be a Boolean function. $R(x,y)$ can be seen as a binary relation $R$ on the set $\bset^n$ with $x\,R \,y \Leftrightarrow R(x,y) = 1$. The \emph{transitive closure} of this relation is the function $R^*(x,y)$ with $R^*(x,y) = 1$ iff there is a sequence $x = x_1,\ldots,x_l = y$ with $R(x_i,x_{i+1}) = 1$ for all $i=1,\ldots,l-1$. \Eg let $R(x,y) = \chi_E(x,y) \vee (\vert x \vert = \vert y \vert)$ be the function that returns $1$ iff there is an edge between $v_{\vert x \vert}$ and $v_{\vert y \vert}$ or $(\vert x \vert = \vert y \vert)$, then is $R^*(x,y) = 1$ iff there the nodes $v_{\vert x \vert}$ and $v_{\vert y \vert}$ are in the same connected component. The transitive closure can be computed implicitly by $O(n^2)$ functional operations using the so called iterative squaring or path doubling technique (see Algorithm \ref{algo:transclos}).

\begin{algorithm}[H]
\caption{$TransitiveClosure(R(x,y))$}
\label{algo:transclos}
\begin{algorithmic}[1]
\REQUIRE Boolean function $R(x,y) \in B_{2n}$
\ENSURE Transitive closure $R^*(x,y)$ of $R(x,y)$
\\
\STATE $R^*(x,y) = R(x,y)$\\
\FOR {$i=0$ to $n$} 
\STATE $R^*(x,y) = \exists z: R^*(x,z) \wedge R^*(z,y)$
\ENDFOR
\RETURN $R^*(x,y)$
\end{algorithmic}
\end{algorithm}
\noindent
Let $O(x,y)$ represent a total order $\prec$ on the input bitstrings, \ie $O(x,y) = 1 \Leftrightarrow x \prec y$ (\eg $O(x,y) = 1 \Leftrightarrow \vert x \vert \leq \vert y \vert$). Since $\prec$ is a total order, the input bitstrings can be sorted in ascending order according to $\prec$. Let $EO(x,l) = 1$ iff $x$ is in the $\vert l \vert$-the position in this sorted sequence. Similar to the transitive closure, it is known (see, \eg \cite{Sawitzki07}) that $EO(x,l)$ can be computed using $O(n^2)$ functional operations (see Algorithm \ref{algo:enumorder}).
\begin{algorithm}[H]
\caption{$EnumarateOrder(O(x,y))$}
\label{algo:enumorder}
\begin{algorithmic}[1]
\REQUIRE Total order $O(x,y) \in B_{2n}$
\ENSURE $EO(x,l)$ with $EO(x,l) = 1$ iff the rank of $x$ is $\vert l \vert$ in the ascending order 
\\
\COMMENT{Compute the pairs of direct successors}\\
\STATE $DS(x,y) = O(x,y) \wedge \overline{\exists z: O(x,z) \wedge O(z,y)}$\\
\COMMENT{$EO_i(x,y,l) = 1$ iff $\vert l \vert \leq 2^i$}\\
\COMMENT{and the distance between the position of $x$ and $y$ is equal to $\vert l \vert$}
\STATE $EO_0(x,y,l) = ((\vert l \vert = 0) \wedge (\vert x \vert = \vert y \vert) \vee ((\vert l \vert = 1) \wedge DS(x,y)))$\\
\COMMENT{Divide and conquer approach: If $2^{i-1} < \vert l \vert \leq 2^i$ then there has to be}\\
\COMMENT{an intermediate bitstring $z$ with distance $2^{i-1}$ to $x$ and $\vert l \vert - 2^{i-1}$ to $y$}\\
\FOR {$i=1$ to $n$}
\STATE $\begin{array}{rcl}
EO_i(x,y,l) & = & ((\vert l \vert \leq 2^{i-1}) \wedge EO_{i-1}(x,y,l)) \vee \left[ (2^{i-1} < \vert l \vert \leq 2^i)\: \wedge \right.\\
& & \left. \exists l_1,z: EO_{i-1}(x,z,2^{i-1}) \wedge EO_{i-1}(z,y,l_1) \wedge (\vert l_1 \vert + 2^{i-1} = \vert l \vert) \right]
\end{array}$
\ENDFOR\\
\COMMENT{Compute the rank according to the distance to the first element}\\
\STATE $EO(x,l) = \exists z: EO_{n}(z,x,l) \wedge \overline{\exists z': O(z',z)}$\\
\RETURN $EO(x,l)$
\end{algorithmic}
\end{algorithm}

\subsection{Interval Graphs}
\label{sec:pre_ig}
We start with a formal definition of (unit) interval graphs.
\begin{definition}[Interval Graph]
Let $\mathcal{I} = \lbrace [a_i,b_i] \mid a_i < b_i \text{ and } 0 \leq i \leq N-1 \rbrace$ be a set of $N \in \mathbb{N}$ intervals on the real line. The \emph{interval graph} $G_{\mathcal{I}} = (V,E)$ has one node for each interval in $\mathcal{I}$ and two nodes $v \neq w$ are adjacent iff the corresponding intervals intersect. If no interval is properly contained in another interval, $G_{\mathcal{I}}$ is called \emph{proper interval graph}. If the length of every interval in $\mathcal{I}$ is equal to $1$ then $G_{\mathcal{I}}$ is called \emph{unit interval graph} (see Fig.\,\ref{fig:intervalgraph}).
\end{definition}
\noindent
Notice that the set of all interval graphs does not change if we restrict ourselves to sets $\mathcal{I}$ where all endpoints are different because if there are two intervals with a shared endpoint then there exists an $\epsilon > 0$ such that moving the shared endpoint of one of the two intervals by $\epsilon$ generates the same interval graph. The definitions of proper and unit interval graphs are equivalent in the sense that they generate the same class of interval graphs \cite{Roberts69}. Hence, in the following we only use the term of unit interval graphs. An undirected graph $H$ is a (unit) interval graph iff there is a set of (unit) intervals $\mathcal{I}$ such that $H = G_\mathcal{I}$. Due to the one-to-one correspondence of the nodes of $G_{\mathcal{I}}$ and the elements of $\mathcal{I}$, we use the notion of node and interval synonymously.

\begin{figure}[ht]
\begin{center}
\begin{tikzpicture}
	\draw[line width=1.5pt] (0,0) to node[sloped,below] {{\footnotesize 0}} (1,0);
	\draw[line width=1.5pt] (0.5,0.5) to node[sloped,below] {{\footnotesize 1}} (1.5,0.5);
	\draw[line width=1.5pt] (0.75, 1) to node[sloped,below] {{\footnotesize 2}} (1.75,1);
	\draw[line width=1.5pt] (1.5,0) to node[sloped,below] {{\footnotesize 3}} (2.5,0);
	\draw[line width=1.5pt] (2,0.5) to node[sloped,below] {\footnotesize{4}} (3,0.5);
	\draw[line width=1.5pt] (2.5,1) to node[sloped,below] {{\footnotesize 5}} (3.5,1);
	\draw[line width=1.5pt] (4,0) to node[sloped,below] {{\footnotesize 6}} (5,0);
	\draw[line width=1.5pt] (4.5,0.5) to node[sloped,below] {{\footnotesize 7}} (5.5,0.5);
%	\draw[line width=1.5pt] (5.5, 0.5) to node[sloped,below] {{\footnotesize 9}} (6.5,0.5);
\end{tikzpicture}
\begin{tikzpicture}
	\node[tnode] (0) at (0,0) {0};
    \node[tnode] (1) at (1,1) {1};
    \node[tnode] (2) at (1,0) {2};
    \node[tnode] (3) at (2,1) {3};
    \node[tnode] (4) at (2,0) {4};
    \node[tnode] (5) at (3,0) {5};
    \node[tnode] (6) at (4,0) {6};
    \node[tnode] (7) at (5,0) {7};
%    \node[tnode] (8) at (6,0) {9};
    
    \draw (0) -- (1);

    \draw (0) -- (2);
    \draw (1) -- (2);
    \draw (1) -- (3);
    \draw (2) -- (3);
    \draw (3) -- (4);
    \draw (3) -- (5);
    \draw (4) -- (5);
    \draw (6) -- (7);
%    \draw (7) -- (8);
\end{tikzpicture}
\end{center}
\caption{Example of a set of unit intervals and the corresponding unit interval graph}
\label{fig:intervalgraph}
\end{figure}
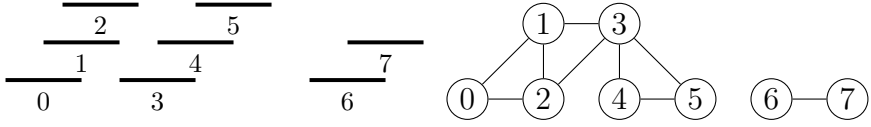

\section{OBDD Size of Interval Graphs}
\label{sec:sizeintervalgraph}

In order to bound the size of a function $f$ using Theorem \ref{thm:minimal_obdd}, we have to count different subfunctions of $f$. We present a way to count the subfunctions of the characteristic function $f := \chi_E$ of the edge set of a graph using the adjacency matrix of the graph. This can give us a better understanding what a subfunction looks like in the graph scenario and get a more graph theoretic approach to subfunctions. 

The adjacency matrix of graphs from special graph classes (\eg for interval graphs) yields some structural properties which we use to bound the size of the $\pi$-OBDD. As we know from the last section, the OBDD size is dependent on the labeling of the nodes. So if we use the knowledge about the structure of the adjacency matrix for a fixed labeling to bound the number of different subfunctions for a variable order $\pi$, we can show an upper and/or lower bound of the $\pi$-OBDD size. 

The rows (columns) of an adjacency matrix correspond to the $x$-variables ($y$-variables) of $\chi_E(x,y)$. We can sort the rows of the adjacency matrix according to a variable order $\pi$ by connecting the $i$-th row to the input $x$ with $\sum_{l=0}^{n-1} x_{\pi(n-l-1)} \cdot 2^l = i$, \ie we let the $l$-th $x$-variable in $\pi$ have significance $2^{n-l-1}$ to sort the rows. This can be done analogously to sort the columns. Thus, the variable order $\pi$ defines a permutation of the rows and columns of the adjacency matrix resulting in a new matrix which we call $\pi$-ordered adjacency matrix.

\begin{definition}
Let $G = (V,E)$ be a graph and $\pi := \pi_{2,n}$ be a $2$-interleaved variable order for the characteristic function $f := \chi_E$. The \emph{$\pi$-ordered} adjacency matrix $A_{\pi}$ of $G$ is defined as follows: $a_{ij} = 1$ iff $f(x,y) = 1$ with $\sum_{l=0}^{n-1} x_{\pi(n-l-1)} \cdot 2^l = i$ and $\sum_{l=0}^{n-1} y_{\pi(n-l-1)} \cdot 2^l = j$. 
\end{definition}
\noindent
Notice that the $\pi$-ordered adjacency matrix is equal to the \enquote{normal} adjacency matrix where the rows and columns are sorted by the node labels iff the variables in $\pi$ are tested with decreasing significance. The $\pi$-ordered adjacency matrix gives us a visualization of the subfunctions in terms of \emph{blocks} of the matrix.

\begin{definition}
Let $n \in \mathbb{N}$ and $A$ be a $2^n \times 2^n$ matrix. For $0 \leq k \leq n$ and $0 \leq i,j \leq 2^k-1$ the \emph{block} $B^k_{ij}$ of $A$ is defined by the quadratic submatrix of size $2^n/2^k \times 2^n/2^k$ which is formed by the intersection of the rows $i \cdot 2^n/2^k, \ldots, (i+1)\cdot 2^n/2^k-1$ and the columns $j \cdot 2^n/2^k, \ldots, (j+1) \cdot 2^n/2^k-1$.
\end{definition}
\noindent
Recall Theorem \ref{thm:minimal_obdd} that we want to count the number of different subfunctions which result from replacing the first $i$ variables according to $\pi$ by constants. We will see later that for an upper bound it is enough to consider only the case when $i$ is even, \ie the number of replaced $x$- and $y$-variables is exactly $i/2$. Now, we can see that the block $B^{i/2}_{\vert a \vert, \vert b \vert}$ represents the function table of the subfunction which results from replacing the $x$-variables by $a \in \bset^{i/2}$ and the $y$-variables by $b \in \bset^{i/2}$. Therefore, counting the number of different blocks $B^{i/2}_{\vert a \vert, \vert b \vert}$ is equivalent to counting the number of different subfunctions.

For instance, let say that the variables are tested with decreasing significance. Then $a_{ij} = 1$ iff $f(x,y) = 1$ with $\sum_{l=0}^{n-1} x_{l} \cdot 2^l = \vert x \vert = i$ and $\sum_{l=0}^{n-1} y_{l} \cdot 2^l = \vert y \vert = j$, \ie the $\pi$-ordered adjacency matrix $A_\pi$ of $G$ is the standard adjacency matrix where the labeling of the columns and rows is ordered by the node labels. In Fig.\,\ref{fig:piordered_matrix} we can see that for every $k$ each subfunction of $f$ where the first $k$ bits (according to $\pi_{2,n}$) are replaced by constants corresponds to a block of this adjacency matrix. 

Bollig and Wegener \cite{BolligW00} use a similar approach to visualize subfunctions of a storage access function by building a matrix whose columns and rows are sorted according to the variable order and correspond to variables (not assignments as in our $\pi$-ordered matrix). Notice that $A_\pi$ is not the communication matrix which is often used to show lower bounds of the OBDD size.

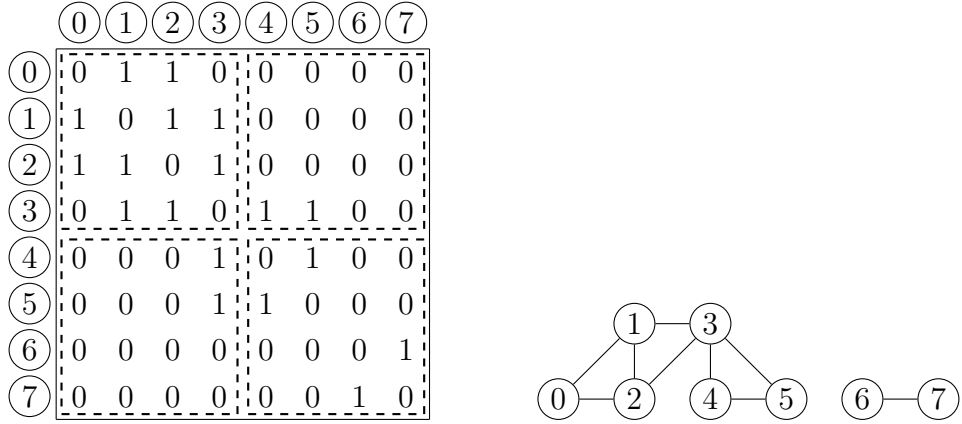
\begin{figure}
\begin{center}
\begin{tikzpicture}
	\node[entry] (00) at (0,0) {0};
    \node[entry,below=0pt of 00] (01) {1};
    \node[entry,below=0pt of 01] (02) {1};
    \node[entry,below=0pt of 02] (03) {0};
    \node[entry,below=0pt of 03] (04) {0};
    \node[entry,below=0pt of 04] (05) {0};
    \node[entry,below=0pt of 05] (06) {0};
    \node[entry,below=0pt of 06] (07) {0};
    
	\node[entry,right=0pt of 00] (10) {1};
    \node[entry,below=0pt of 10] (11) {0};
    \node[entry,below=0pt of 11] (12) {1};
    \node[entry,below=0pt of 12] (13) {1};
    \node[entry,below=0pt of 13] (14) {0};
    \node[entry,below=0pt of 14] (15) {0};
    \node[entry,below=0pt of 15] (16) {0};
    \node[entry,below=0pt of 16] (17) {0};
    
  	\node[entry,right=0pt of 10] (20) {1};
    \node[entry,below=0pt of 20] (21) {1};
    \node[entry,below=0pt of 21] (22) {0};
    \node[entry,below=0pt of 22] (23) {1};
    \node[entry,below=0pt of 23] (24) {0};
    \node[entry,below=0pt of 24] (25) {0};
    \node[entry,below=0pt of 25] (26) {0};
    \node[entry,below=0pt of 26] (27) {0};
    
   	\node[entry,right=0pt of 20] (30) {0};
    \node[entry,below=0pt of 30] (31) {1};
    \node[entry,below=0pt of 31] (32) {1};
    \node[entry,below=0pt of 32] (33) {0};
    \node[entry,below=0pt of 33] (34) {1};
    \node[entry,below=0pt of 34] (35) {1};
    \node[entry,below=0pt of 35] (36) {0};
    \node[entry,below=0pt of 36] (37) {0};
    
   	\node[entry,right=0pt of 30] (40) {0};
    \node[entry,below=0pt of 40] (41) {0};
    \node[entry,below=0pt of 41] (42) {0};
    \node[entry,below=0pt of 42] (43) {1};
    \node[entry,below=0pt of 43] (44) {0};
    \node[entry,below=0pt of 44] (45) {1};
    \node[entry,below=0pt of 45] (46) {0};
    \node[entry,below=0pt of 46] (47) {0};
    
   	\node[entry,right=0pt of 40] (50) {0};
    \node[entry,below=0pt of 50] (51) {0};
    \node[entry,below=0pt of 51] (52) {0};
    \node[entry,below=0pt of 52] (53) {1};
    \node[entry,below=0pt of 53] (54) {1};
    \node[entry,below=0pt of 54] (55) {0};
    \node[entry,below=0pt of 55] (56) {0};
    \node[entry,below=0pt of 56] (57) {0};
        
   	\node[entry,right=0pt of 50] (60) {0};
    \node[entry,below=0pt of 60] (61) {0};
    \node[entry,below=0pt of 61] (62) {0};
    \node[entry,below=0pt of 62] (63) {0};
    \node[entry,below=0pt of 63] (64) {0};
    \node[entry,below=0pt of 64] (65) {0};
    \node[entry,below=0pt of 65] (66) {0};
    \node[entry,below=0pt of 66] (67) {1};
        
   	\node[entry,right=0pt of 60] (70) {0};
    \node[entry,below=0pt of 70] (71) {0};
    \node[entry,below=0pt of 71] (72) {0};
    \node[entry,below=0pt of 72] (73) {0};
    \node[entry,below=0pt of 73] (74) {0};
    \node[entry,below=0pt of 74] (75) {0};
    \node[entry,below=0pt of 75] (76) {1};
    \node[entry,below=0pt of 76] (77) {0}; 
    
    \node[tnode] (c0) at ($ (00.north)+(0,0.35) $) {0}; 
    \node[tnode] (c1) at ($ (10.north)+(0,0.35) $) {1}; 
    \node[tnode] (c2) at ($ (20.north)+(0,0.35) $) {2}; 
    \node[tnode] (c3) at ($ (30.north)+(0,0.35) $) {3}; 
    \node[tnode] (c4) at ($ (40.north)+(0,0.35) $) {4}; 
    \node[tnode] (c5) at ($ (50.north)+(0,0.35) $) {5}; 
    \node[tnode] (c6) at ($ (60.north)+(0,0.35) $) {6}; 
    \node[tnode] (c7) at ($ (70.north)+(0,0.35) $) {7};         

    \node[tnode] (r0) at ($ (00.west)+(-0.35,0) $) {0}; 
    \node[tnode] (r1) at ($ (01.west)+(-0.35,0) $) {1}; 
    \node[tnode] (r2) at ($ (02.west)+(-0.35,0) $) {2}; 
    \node[tnode] (r3) at ($ (03.west)+(-0.35,0) $) {3}; 
    \node[tnode] (r4) at ($ (04.west)+(-0.35,0) $) {4}; 
    \node[tnode] (r5) at ($ (05.west)+(-0.35,0) $) {5}; 
    \node[tnode] (r6) at ($ (06.west)+(-0.35,0) $) {6}; 
    \node[tnode] (r7) at ($ (07.west)+(-0.35,0) $) {7};             
   
    \rectd{00}{03}{33}{30}{2}
    
    \rectd{04}{07}{37}{34}{2}
    
    \rectd{40}{43}{73}{70}{2}    

    \rectd{44}{47}{77}{74}{2}
    
    \draw (00.north west) -- (07.south west);
    \draw (00.north west) -- (70.north east);
    \draw (07.south west) -- (77.south east);
    \draw (70.north east) -- (77.south east);
\end{tikzpicture}
\hspace{1cm}
\begin{tikzpicture}
	\node[tnode] (0) at (0,0) {0};
    \node[tnode] (1) at (1,1) {1};
    \node[tnode] (2) at (1,0) {2};
    \node[tnode] (3) at (2,1) {3};
    \node[tnode] (4) at (2,0) {4};
    \node[tnode] (5) at (3,0) {5};
    \node[tnode] (6) at (4,0) {6};
    \node[tnode] (7) at (5,0) {7};
%    \node[tnode] (8) at (6,0) {9};
    
    \draw (0) -- (1);

    \draw (0) -- (2);
    \draw (1) -- (2);
    \draw (1) -- (3);
    \draw (2) -- (3);
    \draw (3) -- (4);
    \draw (3) -- (5);
    \draw (4) -- (5);
    \draw (6) -- (7);
%    \draw (7) -- (8);
\end{tikzpicture}
\end{center}
\caption{The $\pi$-ordered adjacency matrix of the unit interval graph from Fig.\,\ref{fig:intervalgraph} with $\pi = (2,1,0)$ and framed blocks $B^1_{0,0}, B^1_{1,0}, B^1_{0,1}$ and $B^1_{1,1}$ which correspond to the subfunctions where the first bit is replaced by a constant.}
\label{fig:piordered_matrix}
\end{figure}

Next, we use the $\pi$-ordered adjacency matrix and count the number of different blocks to improve the bounds of the OBDD size of interval graphs. 

\begin{theorem}
\label{thm:ub_sizeintervalgraph}
Let $\pi := \pi_{2,n}$ be the interleaved variable order with decreasing significance and $G = (V,E)$ be an interval graph with $N := \vert V \vert$ nodes. The $\pi$-OBDD size of $\chi_E$ can be bounded above by $O(N \log N)$.
\end{theorem}
\begin{proof}
Let $f := \chi_E$, $1 \leq k \leq n$ and $s_k$ be the number of different subfunctions $f_{\mid \alpha, \beta}$ of $f$ where $\alpha \in \bset^k$ is an assignment to the variables $x_{n-1},\ldots,x_{n-k}$ and $\beta \in \bset^k$ is an assignment to the variables $y_{n-1},\ldots, y_{n-k}$, respectively. The number of different subfunctions where the variables $x_{n-1},\ldots,x_{n-k}$ and $y_{n-1},\ldots, y_{n-k-1}$  are replaced by constants can be bounded by $2 \cdot s_k$ because one additional variable can at most double the number of subfunctions.

We label the nodes according to their position in the sorted sequence of interval starting points (as for example in Fig.\,\ref{fig:intervalgraph}). Recall that the interleaved variable order with decreasing significance means that $a_{i,j}$ is one if and only if interval $i$ intersects interval $j$. Now, notice that if $a_{i,j}$ is zero for $j>i$, \ie interval $j$ has a larger starting point than interval $i$ and does not cut interval $i$, then no interval $j'>j$ with a larger starting point can cut interval $i$. Thus, for every column $i \in \{0,\ldots,N-1\}$, the sequence $(a_{i+1,i},\ldots,a_{N-1,i})$ is zero or starts with a continuous sequence of ones followed by only zeros, \ie\ there exists a $j$ such that $a_{k,i}=1$ for $i < k \le j$ and $a_{k,i}=0$ for $ k > j$.

As seen in the beginning of this section, every subfunction  $f_{\mid \alpha, \beta}$ corresponds to a block of $A_{\pi}$. Let $\beta = 0^k$ and $\vert \alpha \vert \geq 1$, \ie we consider the blocks $B^{k}_{\vert \alpha \vert, 0}$ of size $2^{n-k} \times 2^{n-k}$ (see Fig.\,\ref{fig:IGAdj}). As we observed, every column of $A_{\pi}$ has (below the diagonal) at most one possible \emph{changing} position $k$ such that $a_{k,i} = 1$ and $a_{k+1,i} = 0$. Looking at the sequence $(B^k_{1,0},\ldots,B^k_{2^k-1,0})$ of blocks, this fact implies that a block $B^k_{i,0}$ can only form a new block, \ie all previous block in the sequence are different to this block, if there is a changing position in one column inside of $B^k_{i,0}$ or inside the block $B^k_{i-1,0}$ or between these two blocks. Therefore, every changing position can induce at most two different blocks and, therefore, we can bound the number of different blocks by two times the number of possible changing positions which is at most the number of columns of a block, \ie $2 \cdot 2^{n-k}$. Since the graph is symmetric and the blocks containing the diagonal can only add $2^k$ additional distinct blocks, we can bound the overall number of different blocks by $O(2^{n-k} \cdot 2^k+2^k) = O(2^n)$ and thus $s_k = O(2^n)$. Summing this up over all possible values of $k$ we get $O(2^n \cdot n) = O(N \log N)$ as an upper bound on the size of the $\pi$-OBDD.
\end{proof}

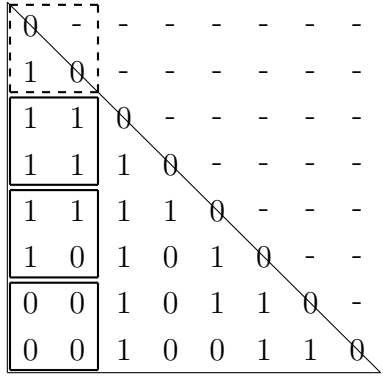
\begin{figure}[t]
\begin{center}
\begin{tikzpicture}
	\node[entry] (00) at (0,0) {0};
    \node[entry,below=0pt of 00] (01) {1};
    \node[entry,below=0pt of 01] (02) {1};
    \node[entry,below=0pt of 02] (03) {1};
    \node[entry,below=0pt of 03] (04) {1};
    \node[entry,below=0pt of 04] (05) {1};
    \node[entry,below=0pt of 05] (06) {0};
    \node[entry,below=0pt of 06] (07) {0};
    
	\node[entry,right=0pt of 00] (10) {-};
    \node[entry,below=0pt of 10] (11) {0};
    \node[entry,below=0pt of 11] (12) {1};
    \node[entry,below=0pt of 12] (13) {1};
    \node[entry,below=0pt of 13] (14) {1};
    \node[entry,below=0pt of 14] (15) {0};
    \node[entry,below=0pt of 15] (16) {0};
    \node[entry,below=0pt of 16] (17) {0};
    
  	\node[entry,right=0pt of 10] (20) {-};
    \node[entry,below=0pt of 20] (21) {-};
    \node[entry,below=0pt of 21] (22) {0};
    \node[entry,below=0pt of 22] (23) {1};
    \node[entry,below=0pt of 23] (24) {1};
    \node[entry,below=0pt of 24] (25) {1};
    \node[entry,below=0pt of 25] (26) {1};
    \node[entry,below=0pt of 26] (27) {1};
    
   	\node[entry,right=0pt of 20] (30) {-};
    \node[entry,below=0pt of 30] (31) {-};
    \node[entry,below=0pt of 31] (32) {-};
    \node[entry,below=0pt of 32] (33) {0};
    \node[entry,below=0pt of 33] (34) {1};
    \node[entry,below=0pt of 34] (35) {0};
    \node[entry,below=0pt of 35] (36) {0};
    \node[entry,below=0pt of 36] (37) {0};
    
   	\node[entry,right=0pt of 30] (40) {-};
    \node[entry,below=0pt of 40] (41) {-};
    \node[entry,below=0pt of 41] (42) {-};
    \node[entry,below=0pt of 42] (43) {-};
    \node[entry,below=0pt of 43] (44) {0};
    \node[entry,below=0pt of 44] (45) {1};
    \node[entry,below=0pt of 45] (46) {1};
    \node[entry,below=0pt of 46] (47) {0};
    
   	\node[entry,right=0pt of 40] (50) {-};
    \node[entry,below=0pt of 50] (51) {-};
    \node[entry,below=0pt of 51] (52) {-};
    \node[entry,below=0pt of 52] (53) {-};
    \node[entry,below=0pt of 53] (54) {-};
    \node[entry,below=0pt of 54] (55) {0};
    \node[entry,below=0pt of 55] (56) {1};
    \node[entry,below=0pt of 56] (57) {1};
        
   	\node[entry,right=0pt of 50] (60) {-};
    \node[entry,below=0pt of 60] (61) {-};
    \node[entry,below=0pt of 61] (62) {-};
    \node[entry,below=0pt of 62] (63) {-};
    \node[entry,below=0pt of 63] (64) {-};
    \node[entry,below=0pt of 64] (65) {-};
    \node[entry,below=0pt of 65] (66) {0};
    \node[entry,below=0pt of 66] (67) {1};
        
   	\node[entry,right=0pt of 60] (70) {-};
    \node[entry,below=0pt of 70] (71) {-};
    \node[entry,below=0pt of 71] (72) {-};
    \node[entry,below=0pt of 72] (73) {-};
    \node[entry,below=0pt of 73] (74) {-};
    \node[entry,below=0pt of 74] (75) {-};
    \node[entry,below=0pt of 75] (76) {-};
    \node[entry,below=0pt of 76] (77) {0};  
    
%    \draw[dashed,thick] (00.north west)+(0.5pt,0) -- (01.south west)+(0.5pt,0);
%    \draw[dashed,thick] (01.south west)+(0.5pt,0) -- (11.south east)+(-0.5pt,0);
%    \draw[dashed,thick] (11.south east)+(-0.5pt,0) -- (10.north east)+(-0.5pt,0);
%    \draw[dashed,thick] (10.north east)+(-0.5pt,0) -- (00.north west)+(0.5pt,0);
    
    \rectd{00}{01}{11}{10}{1}
    
    \rect{02}{03}{13}{12}{1}
    
    \rect{04}{05}{15}{14}{1}    

    \rect{06}{07}{17}{16}{1}
    
    \draw (00.north west) -- (77.south east);
    \draw (00.north west) -- (07.south west);
    \draw (07.south west) -- (77.south east);
\end{tikzpicture}
\end{center}
\caption{Possible adjacency matrix with $8$ nodes and framed subfunctions $f_{\mid \alpha, \beta}$ with $\beta = 0^k$, $\vert \alpha \vert \geq 1$, and $k=2$.}
\label{fig:IGAdj}
\end{figure}
\noindent
In the case of unit interval graph Nunkesser and Woelfel \cite{NuWo09} proved that the OBDD size is bounded below by $\Omega(N/\log N)$ and above by $O(N/\sqrt{\log N})$. We can close this gap by using the $\pi$-ordered adjacency matrix to get a better upper bound on the number of subfunctions of $\chi_E$ for large values of $k$.

\begin{theorem}
\label{thm:sizeunitintervalgraph}
Let $\pi$ be the interleaved variable order with decreasing significance. The $\pi$-OBDD size of $\chi_E$ for a unit interval graph $G = (V,E)$ is $O(N / \log N)$.
\end{theorem}
\begin{proof}
Again, let $f := \chi_E$ and $s_k$ be the number of different subfunctions $f_{\mid \alpha, \beta}$ of $f$ where $\alpha$ is an assignment to the variables $x_{n-1},\ldots,x_{n-k}$ and $\beta$ is an assignment to the variables $y_{n-1},\ldots, y_{n-k}$ respectively. As we have seen in the last proof, the number of different subfunctions where the first $k$ $x$-variables and $k+1$ $y$-variables are replaced by constants can be bounded by $2 \cdot s_k$.

We label the nodes according to their interval starting points. We know that $f_{\mid \alpha, \beta}$ corresponds to the block $B^k_{\vert \alpha \vert,\vert \beta \vert}$ of the $\pi$-ordered adjacency matrix of $G$. Let $\vert \alpha \vert > \vert \beta \vert$. Every column of these blocks consists of a beginning sequence of ones of length $l \geq 0$ and an ending sequence of zeros of length $C-l$, where $C = 2^{n-k}$ is the number of rows and columns of $B^k_{\vert \alpha \vert, \vert \beta \vert}$. Let $l_1,\ldots,l_C$ be the lengths of the beginning sequence consisting of ones of every column in the block $B^k_{\vert \alpha \vert, \vert \beta \vert}$. Recall that the intervals are labeled according to their interval starting point. Since $G$ is a unit interval graph, this is equivalent to labeling them according to their interval ending points, \ie if $j>i$, then interval $j$ starts \emph{and} ends after interval $i$. This implies that the sequence $l_1, \ldots, l_C$ is monotonically increasing, \ie $l_1 \leq l_2 \leq \ldots \leq l_C$. How many different blocks of this form can be constructed? We can construct such a block by drawing $C$ numbers between $0$ and $C$ and sorting them, \ie it is equivalent to selecting $C$ numbers out of $\lbrace 0,\ldots,C \rbrace$ with repetition, where order does not matter. The number of $C$-combinations with repetition is equal to $\binom{(C+1)+C-1}{C} = \binom{2C}{C}$ and this can be bounded above by $2^{2C}$. Since $G$ is symmetric, this is also a bound on the number of different blocks above the diagonal. 
The omitted blocks on the diagonal can be constructed in a similar way: At first, the diagonal of these blocks is zero and the blocks are symmetric. Below the diagonal the blocks also consist of a sequence of ones probably followed by a sequence of zeros. So the number of different blocks is bounded above by the number of different blocks, which are not on the diagonal, \ie by $2^{2C}$. Hence, for $C = 2^{n-k}$ we can bound $s_k$ above by $3 \cdot 2^{2^{n-k+1}}$. 
Nunkesser and Woelfel \cite{NuWo09} also showed that $s_k \leq 2^{k+2}-2$. Therefore, the OBDD size is at most
$$ \begin{array}{rcl}
\sum\limits_{k=0}^{n-1} \min \lbrace 2^{k+2}, 3 \cdot 2^{2^{n-k+1}} \rbrace & \leq & \sum\limits_{k=0}^{n-\log n +1} 2^{k+2} + 3 \cdot \sum\limits_{k=n-\log n +2}^{n-1} 2^{2^{n-k+1}}\\
& \leq & 2^{n-\log n +4} + 3 \cdot 2^{2^{\log n -1}}\cdot (\log n - 2)\\
& = & O(N/\log N) + O(\sqrt{N} \cdot \log \log N)\\
& = & O(N/\log N).
\end{array}$$
\end{proof}
\noindent
The difference between unit and general interval graphs is that in general interval graphs there is no dependence between the columns of the $\pi$-ordered adjacency matrix, which is important for our lower bound, while in unit interval graphs, the row number of the last $1$ entry in a column is increasing from left to right. 

The proofs of the upper bounds suggest that the number of blocks $B_{i,j}^k$ with a changing position roughly determines the number of OBDD nodes labeled by $x_{n-k-1}$. We know that every layer of the OBDD, \ie every set of OBDD nodes labeled by the same variable, has size $O(N)$ which means that there has to be $\Omega(n)$ layers of the OBDD of size $\Omega(N)$ to show a lower bound of $\Omega(N \log N)$. Explicitly constructing a worst-case interval graph with OBDD size of $\Omega(N \log N)$ is difficult because $\Omega(n)$ layers correspond to $\Omega(n)$ values of $k$ and, since the block $B_{i,j}^k$ results from dividing a block $B_{i',j'}^{k-1}$, many dependencies have to be considered to ensure that $\Omega(N)$ blocks are different for all the possible values of $k$. 

In order to overcome these dependencies, we look at a random interval graph and compute the expected value of the number of different blocks for $\Omega(n)$ values of $k$. 
Intuitively, in the worst-case the lengths of the $1$-sequences of the columns are uniformly distributed such that there are many blocks with a small number of changing positions inside which maximizes the possibility that there are many different blocks. Choosing an appropriate distribution on the set of interval graphs, we show that the expected number of different blocks with one changing position is $\Omega(N)$ for $\Omega(n)$ values of $k$. Due to the linearity of expectation, the expected value of the sum of the number of different blocks over all values of $k$ is $\Omega(N n ) = \Omega(N \log N)$, \ie there is an interval graph whose OBDD size is also $\Omega(N \log N)$.
\begin{figure}
\begin{center}
\begin{tikzpicture}
	\node[entry] (00) at (0,0) {0};
    \node[entry,below=0pt of 00] (01) {1};
    \node[entry,below=0pt of 01] (02) {1};
    \node[entry,below=0pt of 02] (03) {1};
    \node[entry,below=0pt of 03] (04) {};
    \node[entry,below=0pt of 04] (05) {};
    \node[entry,below=0pt of 05] (06) {};
    \node[entry,below=0pt of 06] (07) {};
    
	\node[entry,right=0pt of 00] (10) {1};
    \node[entry,below=0pt of 10] (11) {0};
    \node[entry,below=0pt of 11] (12) {1};
    \node[entry,below=0pt of 12] (13) {1};
    \node[entry,below=0pt of 13] (14) {};
    \node[entry,below=0pt of 14] (15) {};
    \node[entry,below=0pt of 15] (16) {};
    \node[entry,below=0pt of 16] (17) {};
    
    \node[entry] (R1) at (0.95,-3.4) {{\LARGE $R$}};
    
  	\node[entry,right=0pt of 10] (20) {1};
    \node[entry,below=0pt of 20] (21) {1};
    \node[entry,below=0pt of 21] (22) {0};
    \node[entry,below=0pt of 22] (23) {1};
    \node[entry,below=0pt of 23] (24) {};
    \node[entry,below=0pt of 24] (25) {};
    \node[entry,below=0pt of 25] (26) {};
    \node[entry,below=0pt of 26] (27) {};
    
   	\node[entry,right=0pt of 20] (30) {1};
    \node[entry,below=0pt of 30] (31) {1};
    \node[entry,below=0pt of 31] (32) {1};
    \node[entry,below=0pt of 32] (33) {0};
    \node[entry,below=0pt of 33] (34) {};
    \node[entry,below=0pt of 34] (35) {};
    \node[entry,below=0pt of 35] (36) {};
    \node[entry,below=0pt of 36] (37) {};
    
   	\node[entry,right=0pt of 30] (40) {};
    \node[entry,below=0pt of 40] (41) {};
    \node[entry,below=0pt of 41] (42) {};
    \node[entry,below=0pt of 42] (43) {};
    \node[entry,below=0pt of 43] (44) {0};
    \node[entry,below=0pt of 44] (45) {1};
    \node[entry,below=0pt of 45] (46) {1};
    \node[entry,below=0pt of 46] (47) {1};
    
    \node[entry] (R2) at (3.4,-0.85) {{\LARGE $R^T$}};
    
   	\node[entry,right=0pt of 40] (50) {};
    \node[entry,below=0pt of 50] (51) {};
    \node[entry,below=0pt of 51] (52) {};
    \node[entry,below=0pt of 52] (53) {};
    \node[entry,below=0pt of 53] (54) {1};
    \node[entry,below=0pt of 54] (55) {0};
    \node[entry,below=0pt of 55] (56) {1};
    \node[entry,below=0pt of 56] (57) {1};
        
   	\node[entry,right=0pt of 50] (60) {};
    \node[entry,below=0pt of 60] (61) {};
    \node[entry,below=0pt of 61] (62) {};
    \node[entry,below=0pt of 62] (63) {};
    \node[entry,below=0pt of 63] (64) {1};
    \node[entry,below=0pt of 64] (65) {1};
    \node[entry,below=0pt of 65] (66) {0};
    \node[entry,below=0pt of 66] (67) {1};
        
   	\node[entry,right=0pt of 60] (70) {};
    \node[entry,below=0pt of 70] (71) {};
    \node[entry,below=0pt of 71] (72) {};
    \node[entry,below=0pt of 72] (73) {};
    \node[entry,below=0pt of 73] (74) {1};
    \node[entry,below=0pt of 74] (75) {1};
    \node[entry,below=0pt of 75] (76) {1};
    \node[entry,below=0pt of 76] (77) {0};

%    \rectd{00}{03}{33}{30}{2}
    
    \rectd{04}{07}{37}{34}{2}
    
    \rectd{40}{43}{73}{70}{2}    

%    \rectd{44}{47}{77}{74}{2}
    
    \draw (00.north west) -- (07.south west);
    \draw (00.north west) -- (70.north east);
    \draw (07.south west) -- (77.south east);
    \draw (70.north east) -- (77.south east);
\end{tikzpicture}
\end{center}
\caption{Random interval graph where only $R$ is generated randomly}
\label{fig:randig}
\end{figure}

\begin{theorem}
The worst-case $\pi$-OBDD size of an interval graph is $\Omega(N \log N)$ where the nodes are labeled according to the interval starting points and $\pi$ is an interleaved variable order with decreasing significance.
\end{theorem}
\begin{proof}
We describe a random process to generate an interval graph where the adjacency matrix is constant except the $N/2 \times N/2$ lower left submatrix which we denote by $R$ (see Fig.\,\ref{fig:randig}). For this, we choose the length of the $1$-sequence of column $j$ for all $0 \leq j \leq N/2-1$ uniformly at random from $\lbrace N/2-j,\ldots, N-1-j \rbrace$ and for all $N/2 \leq j \leq N-1$ the length of column $j$ is equal to $N-1-j$. As a result, the length of the $1$-sequence of each column within $R$ is uniform at random in $\lbrace 1, \ldots, N/2 \rbrace$. 

Let $G = (V,E)$ be a random interval graph generated by the above process and $f := \chi_E$. Let $1 \leq k \leq n$ and $s_k$ be the number of different subfunctions $f_{\mid \alpha, \beta}$ of $f$ where $\alpha \in \bset^k$ is an assignment to the variables $x_{n-1},\ldots, x_{n-k}$ and $\beta \in \bset^k$ is an assignment to the variables $y_{n-1},\ldots, y_{n-k}$, respectively, and $f_{\mid \alpha, \beta}$ is essentially dependent on $x_{n-k-1}$. We show that the expected value of $s_k$ with $n/2+1 \leq k \leq (3/4)n$ is $\Omega(2^n) = \Omega(N)$. Therefore, there has to be an interval graph with $\pi$-OBDD size $\Omega(N \log N)$.

We known that $k$ induces a grid in $R$ consisting of $2^{n-k} \times 2^{n-k}$ blocks. At first, we calculate the expected number of blocks with exactly one changing position. The probability that a fixed block of size $L \times L$ with $L \leq 2^{n/2-1}$ has exactly one changing position is  
$$ \begin{array}{rcl}
\sum\limits_{i=1}^{L} \dfrac{L-1}{2^{n-1}} \cdot \left(1-\dfrac{L-1}{2^{n-1}}\right)^{L-1} & \geq & \dfrac{L \cdot (L-1)}{2^{n-1}} \cdot \left(1-\dfrac{L}{2^{n-1}}\right)^{L-1}\\
& \geq & \dfrac{L \cdot (L-1)}{2^{n-1}} \cdot \left(1-\dfrac{2^{n/2}}{2^{n-1}}\right)^{2^{n/2-1}-1}\\
& \geq & \dfrac{L \cdot (L-1)}{2^{n-1}} \cdot e^{-1}.
\end{array} $$
Let $n/2+1 \leq k \leq (3/4)n$ be fixed. Since we have $2^{k-1}\cdot 2^{k-1}$ blocks of size $2^{n-k} \times 2^{n-k}$ in $R$, the expected value of the number of blocks with exactly one changing position is at least
$$ \dfrac{1}{2e} \cdot 2^k \cdot 2^k \cdot \dfrac{2^{n-k}\cdot (2^{n-k}-1)}{2^{n}} = \dfrac{1}{2e} \cdot 2^k \cdot (2^{n-k}-1) = \Omega(2^n). $$
Now, we have to ensure that these blocks correspond to different subfunctions which are also essentially dependent on $x_{n-k-1}$. The subfunctions, where, additionally, $x_{n-k-1}$ is replaced by $0$ and $1$, correspond to a half of the blocks. Thus, a block is symmetric iff the corresponding subfunction is not essentially dependent on $x_{n-k-1}$. Due to the one changing position in each block, this is not possible. Blocks $B_{i,j}^k$ and $B_{i',j}^k$ with exactly one changing position and $i \neq i'$ clearly correspond to different subfunctions because they are in the same block column. But blocks $B_{i,j}^k$ and $B_{i',j'}^k$ with $j \neq j'$, \ie from different block columns, do not have to be different. By replacing some columns of the matrix by constants, we ensure that this also holds.

Consider the case $k = (3/4)n$, \ie the finest grid of $R$ made by $2^{n-k} \times 2^{n-k}$ blocks with $n/2+1 \leq k \leq (3/4)n$. For every block column $0 \leq j \leq 2^k-1$ we fix the first $k$ columns of $B_{i,j}^k$ with $0 \leq i \leq 2^k-1$ such that they represent the binary number $[j]_2$ of the column index. Thus, we have that blocks $B_{i,j}^k$ and $B_{i',j'}^k$ with $j \neq j'$ are always different. Since we looked at the finest grid, this also holds for smaller values of $k$ because every larger block is equal to a union of small blocks. The probability that a block contains exactly one changing position is smaller than before, since we fix some columns. For $k = (3/4)n$ the number of fixed columns is $(3/4)n$ and in each $k \rightarrow k-1$ step this number is doubled, \ie for $n/2+1 \leq k \leq (3/4)n$ the number of \enquote{free} columns is
$$ 2^{n-k}- 2^{(3/4)n-k} \cdot (3/4)n = 2^{n-k} -2^{(3/4n-k+\log((3/4)n))} = \Omega(2^{n-k}) $$
for $n$ large enough. Replacing $L = 2^{n-k}$ by $\Omega(2^{n-k})$ in the calculation of the expectation does not change the asymptotic behavior. Thus, the expected number of blocks with exactly one changing position remains $\Omega(2^n)$ for every $n/2+1 \leq k \leq (3/4)n$.
\end{proof}

\section{Implicit Algorithms on Interval Graphs}
\label{sec:algig}
In this section, we want to develop a maximum matching algorithm on unit interval graphs and a coloring algorithm on unit and general interval graphs. Before we start with the algorithms, we have to investigate a special function class, which we will use in our algorithms, so called multivariate threshold functions. This function class was investigated in \cite{Woe2006} to analyze the running time of an implicit topological sorting algorithm on grid graphs and Woelfel \cite{Woe2006} looked into the OBDD size of these functions for the interleaved variable order with increasing significance, \ie just the reverse of our variable order. Hosaka et al. \cite{HosakaTKY97} showed that the difference of the OBDD sizes for this two orders is at most $n-1$. We can show that an OBDD using our variable order is not only small but can also be constructed efficiently which is important in view of the implementation.

\subsection{Constructing OBDDs for Multivariate Threshold Functions}
\label{sec:algig_thresh}

We start with a definition of multivariate threshold functions \cite{Woe2006}.

\begin{definition}
A Boolean function $f: \bset^{kn} \rightarrow \bset$ with $k$ input variable vectors $x^{(1)},\ldots,x^{(k)} \in \bset^n$ of length $n$ is called $k$-variate threshold function, if there exist a threshold $T \in \mathbb{Z}$ and $W \in \mathbb{N}$ and weights $w_1,\ldots,w_k \in \lbrace -W, \ldots, W \rbrace$ such that
$$ f(x^{(1)}, \ldots, x^{(k)}) = 1 \Leftrightarrow \sum\limits_{j=1}^k w_j \cdot \vert x^{(j)} \vert \geq T. $$
The set of $k$-variate threshold functions $f \in B_{kn}$ with weight parameter $W$ is denoted by $\mathbb{T}^W_{k,n}$.
\end{definition}
\noindent
Woelfel \cite{Woe2006} showed that there exists an OBDD representing a multivariate threshold function $f \in \mathbb{T}^W_{k,n}$ of size $O(k^2 W n)$ and such an OBDD can be constructed efficiently. Our proof for our variable order is similar to the proof in \cite{Woe2006}: It is sufficient to look at the carry values of the sum $\sum_{j=1}^k w_j \cdot \vert x^{(j)} \vert - T$ and, especially, at the carry value generated at the position with the most significance. Reading the bits with increasing significance, Woelfel showed that after each bit it is enough to store a number with absolute value $O(kW)$ to compute the carry values. Here, we show that the influence of input bits with lower significance is small such that we can also bound the number which we have to store after each bit while we read the bits with decreasing significance.

\begin{theorem}
\label{thm:obddthresholdfunctions}
Let $f \in \mathbb{T}^W_{k,n}$ be a $k$-variate threshold function with weight bound $W \in \mathbb{N}$ and $\pi_{k,n}$ be the $k$-interleaved variable order where the variables are tested with decreasing significance. Then we can construct a $\pi_{k,n}$-OBDD representing $f$ with width $O(kW)$ and size $O(k^2Wn)$ in time $O(k^2 W n)$.
\end{theorem}
\begin{proof}
Similar to the proof in \cite{Woe2006}, we choose $T_0,\ldots,T_{n-1} \in \bset$ and $T_n \in \mathbb{Z}$ such that $-T = \sum_{i=0}^{n} T_i \cdot 2^i$. Notice that the $T_i$ are unique, and that $T_0,\ldots,T_{n-1}$ are the $n$ least significant coefficients of $\vert T \vert$ in binary representation and $T_n$ is the number of times that we need to add $2^n$ in order to make up for the missing coefficients in this binary representation. The function value of $f$ is determined by the sign of $S := -T + \sum_{j=1}^k w_j \cdot \vert x^{(j)} \vert = \sum_{i=0}^{n-1} (T_i + \sum_{j=1}^k w_j \cdot x^{(j)}_i) \cdot 2^i + T_n \cdot 2^n$. 
Now, we represent $S$ in the same way as $T$, \ie we define $S_0,\ldots,S_{n-1}\in \{0,1\}$ and $S_n \in \mathbb{Z}$ as the unique coefficients satisfying $S = \sum_{i=0}^{n} S_i \cdot 2^i$. 
We want to compute $S_i$ step-by-step: Notice that $S_i$ results from adding $T_i+\sum_{j=1}^k w_j \cdot x^{(j)}_i$ and the carry value which is generated at position $i-1$, and taking the remainder of the division of this sum by two. In particular, $S_i$ is only influenced by factors of $2^j$ for $j\le i$, and it holds that for $0\le i\le n-1$
$$ \begin{array}{ll} 
S_i := \left( c_{i-1}+T_i+\sum\limits_{j=1}^k w_j x^{(j)}_i \right) \mod2 2 & \text{ and }\\
c_i := \left\lfloor \left( c_{i-1} + T_i + \sum\limits_{j=1}^k w_j x^{(j)}_i \right) / 2 \right\rfloor &
\end{array} $$
with $c_{-1} = 0$. Finally, we compute $S_n = c_{n-1}+T_n$. Now, we have $f(x^{(1)}, \ldots, x^{(k)}) = 1 \Leftrightarrow c_{n-1} \geq -T_n$, \ie it is sufficient to compute the $c_i$ values.

We rewrite the $c_i$ to have them in a more convenient form. Notice that for $m,n \in \mathbb{N}$ and $x \in \mathbb{R}$ it holds that $\left\lfloor \dfrac{\lfloor x \rfloor + m}{n} \right\rfloor = \left\lfloor \dfrac{x + m}{n} \right\rfloor$ (see, \eg \cite{Knuth}). So we have

\begin{eqnarray*}
c_1 & = & \left\lfloor \left( c_{0} + T_1 + \sum\limits_{j=1}^k w_j x^{(j)}_1 \right) / 2 \right\rfloor\\
    & = & \left\lfloor \left( \left\lfloor \left( T_0 + \sum\limits_{j=1}^k w_j x^{(j)}_0 \right) / 2 \right\rfloor + T_1 + \sum\limits_{j=1}^k w_j x^{(j)}_1 \right) / 2 \right\rfloor\\
    & = & \left\lfloor \left( \left( T_0 + \sum\limits_{j=1}^k w_j x^{(j)}_0 \right) / 2 + T_1 + \sum\limits_{j=1}^k w_j x^{(j)}_1 \right) / 2 \right\rfloor\\
    & = & \left\lfloor \left( T_0 + \sum\limits_{j=1}^k w_j x^{(j)}_0 \right) / 4 + \left( T_1 + \sum\limits_{j=1}^k w_j x^{(j)}_1 \right) / 2 \right\rfloor.
\end{eqnarray*}
Let $c'_i = \frac{T_i + \sum\limits_{j=1}^k w_j x^{(j)}_i}{2^{n-i}}$. Applying the above observation iteratively, we have
$$ c_{n-1} = \left\lfloor \sum\limits_{i=0}^{n-1} c'_i \right\rfloor. $$

According to our variable order, we have to compute $c_i'$ backwards from $n-1$ to $0$. This is possible because each $c_i'$ only depends on $i$. We describe an algorithm that is divided into phases. In each phase, the algorithm is in a state $Q$, reads all $k$ bits of the input variable vectors of the same significance and changes the state depending on the former state and the read bits. After phase $i$, the algorithm has the correct sum of the summands from $n-1$ to $i$. Notice that the bits with lesser significance can only add a value to $S$ with bounded absolute value, so if the accumulated sum has a large enough absolute value, then we can already decide which sign $S$ has.

Let us start with phase $n-1$ and state $Q = 0$. In phase $1 \leq i \leq n-1$ we compute the value of $c'_i$ by reading $x^{(1)}_i,\ldots,x^{(k)}_i$:
\begin{enumerate}
\item If $c'_i+Q \geq -T_n + (kW+1)/2^{n-i}$ then change into the accepting state $Q_{acc}$.
\item If $c'_i+Q < -T_n - (kW+1)/2^{n-i}$ then change into the rejecting state $Q_{rej}$.
\item Otherwise update the state $Q = c'_i+Q$ and go to phase $i-1$.
\end{enumerate}
In phase $0$ we compute $c'_0$ and accept iff $\lfloor c'_0 + Q \rfloor \geq -T_n$.

If we reach phase $0$ then the output is correct due to our above observations. So we have to show that we correctly accept/reject within phase $i$ with $1 \leq i \leq n-1$. For $i = 0, \ldots, n-1$ it is $\vert c'_i \vert \leq \frac{kW+1}{2^{n-i}}$ because $T_i \in \bset$ and all weights are bounded by $W$ and therefore
$$ \left\vert \sum\limits_{l=0}^{i} c'_l \right\vert \leq \sum\limits_{l=0}^{i} \dfrac{kW+1}{2^{n-l}} = \dfrac{kW+1}{2^{n-i}}\sum\limits_{l=0}^{i} \dfrac{1}{2^{l}} \leq \dfrac{kW+1}{2^{n-i}} \cdot 2 = \dfrac{kW+1}{2^{n-i-1}}.$$
\Ie if in phase $i$ it is either $c'_i + Q \geq -T_n + (kW+1)/2^{n-i}$ or $c'_i + Q < -T_n - (kW+1)/2^{n-i}$, then we know that $\left\lfloor \sum\limits_{i=0}^{n-1} c'_i \right\rfloor \geq -T_n$ or $\left\lfloor \sum\limits_{i=0}^{n-1} c'_i \right\rfloor < -T_n$ respectively. So our algorithm works correctly.

Based on this algorithm the construction of the $\pi_{k,n}$-OBDD is easy: Assume that we update the state immediately after reading an input variable. Then each state is represented by an OBDD node labeled by the variable which the algorithm will read next. The states for accepting and rejecting are represented by the sinks. The edges correspond to the state transition of the algorithm. If we are not in an accepting or rejecting state, we know that the state value is between $-T_n-\frac{kW+1}{2^{n-i-1}}$ and $-T_n+\frac{kW+1}{2^{n-i-1}}-1$. We also know that $\vert c'_i \vert \leq \frac{kW+1}{2^{n-i}}$, \ie the values computed in phase $i$ has to be between
$$ -T_n-\frac{kW+1}{2^{n-i-1}}-\frac{kW+1}{2^{n-i}} \text{ and } -T_n+\frac{kW+1}{2^{n-i-1}}-1+\frac{kW+1}{2^{n-i}}.$$
So all values are in the interval $I = [-T_n-\frac{3}{2}\frac{kW+1}{2^{n-i-1}}, -T_n+\frac{3}{2}\frac{kW+1}{2^{n-i-1}})$. The denominator of $c'_i$ is an integer, \ie only at most $2^{n-i} \cdot \vert I \vert = O(kW)$ values of $I$ are possible during the computation. Therefore, we have an OBDD width of $O(kW)$ and overall an OBDD size of $O(k^2Wn)$. The construction algorithm is straightforward and has a running time which is linear to the OBDD size.
\end{proof}
\noindent
The proof of Theorem \ref{thm:obddthresholdfunctions} also showed that the \emph{complete}-OBDD width is bounded by $O(kW)$ where an OBDD is called complete if the length of every path from the root to a sink is equal to the number of input bits, \ie all variables are tested on the path. A binary synthesis of two functions with complete-OBDD width $w_1$ and $w_2$ has a complete-OBDD width of at most $w_1 \cdot w_2$ \cite{SawitzkiLATIN06}. Since the complete-OBDD size is an upper bound on the general OBDD size, we can compute a sequence of $O(1)$ binary synthesis of multivariate threshold functions efficiently using the interleaved variable order with decreasing significance if $k$ and $W$ are constants.

We use the arithmetic notation in our algorithm instead of the functional notation whenever we use multivariate threshold functions or simple combination of multivariate threshold functions, \eg we denote by $\vert x \vert - \vert y \vert = 1$ the conjunction of the multivariate threshold functions $f(x,y) = 1 \Leftrightarrow \vert x \vert - \vert y \vert \geq 1$ and $g(x,y) = 1 \Leftrightarrow \vert y \vert - \vert x \vert \geq -1$. 

\subsection{Maximum Matching on Unit Interval Graphs}
\label{sec:algig_matching}
Let $G = (V,E)$ be a unit interval graph and the nodes are labeled according to the sorted sequence of starting points. Our maximum matching algorithm is based on a simple observation that was also used in a parallel algorithm for this problem \cite{ChungPC97}: Assume that the unit interval graph is connected (otherwise this observation holds for every connected component). Then we have $\lbrace v_i, v_{i+1} \rbrace \in E$ for $i = 0, \ldots, N-2$. Assume that there is an $i \in \lbrace 0, \ldots, N-2 \rbrace$ such that $\lbrace v_i, v_{i+1} \rbrace \not\in E$, then due to the connectivity there has to be another interval with starting point left of $v_i$ or right of $v_{i+1}$, which intersects both intervals $v_i$ and $v_{i+1}$. The length of this interval would be larger than $1$ which is a contradiction.

Algorithm \ref{algo:unitmatching} uses besides the characteristic function $\chi_E$ also the characteristic function of the set of nodes. This is important if the number of nodes is not a power of two and we have assignments to the input variables which do not represent a node. Since we label our node according to their interval starting point, we have that the characteristic function of the node set is equal to $f(x) = 1 \Leftrightarrow \vert x \vert < N$.

\begin{algorithm}[H]
\caption{Implicit maximum matching algorithm for unit interval graphs}
\label{algo:unitmatching}
\begin{algorithmic}[1]
\REQUIRE Unit interval graph $\chi_E$
\ENSURE Matching $\chi_M$
\\ \COMMENT {Compute path graph}
\STATE $\chi_{\overrightarrow{E}}(x,y) = \chi_E(x,y) \wedge (\vert y \vert - \vert x \vert = 1)$\\
\COMMENT {Compute set of starting nodes}
\STATE $First(z) = (\vert z \vert < N) \wedge \forall x: \overline{\chi_{\overrightarrow{E}}(x,z)}$
\\ \COMMENT {Compute set of reachable nodes}
\STATE $S(z) = \exists z': \chi_{\overrightarrow{E}}(z,z')$
\STATE $Reachable(x,y) = (\vert x \vert \leq \vert y \vert) \wedge \forall z: (\vert x \vert \leq \vert z \vert < \vert y \vert) \Rightarrow S(z)$
\STATE $Reachable(x,y) = Reachable(x,y) \wedge (\vert x \vert < N) \wedge (\vert y \vert < N)$
\\ \COMMENT {Compute matching}
\STATE {$F(x) = \exists z,d: First(z) \wedge Reachable(z,x) \wedge (\vert x \vert-\vert z \vert = 2 \vert d \vert)$}
\STATE {$M(x,y) = \chi_{\overrightarrow{E}}(x,y) \wedge F(x)$}
\STATE {$\chi_M(x,y) = M(x,y) \vee M(y,x)$}
\RETURN $\chi_M$
\end{algorithmic}
\end{algorithm}
\noindent
At first, the algorithm computes a directed path graph, \ie a union of paths, which is a subgraph of the input graph and consists of the edges $(x,y)$ with $\vert x \vert - \vert y \vert = 1$. As we have seen, for every connected component this path consists of all nodes within the component. Maximum matchings on vertex disjoint paths can be computed with $O(\log^2 N)$ functional operations \cite{BolligP12}. Here, we know that every path $P$ consists of a consecutive sequence of nodes, \ie $P = (v_i, \ldots, v_{k})$ for $0 \leq i \leq k \leq N-1$. We can use this information to lower the number of functional operations: We compute the set of nodes which are starting nodes of the paths. Then we want to compute the connected components of the graph. Usually, this is done by computing the transitive closure, which needs $O(\log^2 N)$ operations. Here, we can do it better: Two nodes $x$ and $y$ of the unit interval graph are connected iff every node $z$ with $\vert x \vert \leq \vert z \vert < \vert y \vert$ has a successor, \ie there is an edge $(v_{\vert z \vert}, v_{\vert z \vert +1}) \in \overrightarrow{E}$. Having this information, we can compute the matching by adding every second edge of a path to the matching beginning with the first edge. To compute this set of edges on general paths, the distance of every node to the first node has to be computed which can be done by an iterative squaring approach with $O(\log^2 N)$ functional operations \cite{BolligP12}. Here, we can easily determine the set of edges by comparing the difference of two node labels due to the structure of the paths.

\begin{theorem}
Algorithm \ref{algo:unitmatching} computes a maximum matching for unit interval graphs using $O(\log N)$ functional operations.
\end{theorem}
\begin{proof}
As we have seen in the beginning of this section, every connected component has always a path which consists of consecutive nodes and visits every node in this component. The algorithm computes such paths and construct a maximum matching of each path. Clearly, the union of these matchings is a maximum matching of the complete graph.

The number of functional operations is determined by the lines $2$-$4$ where we use quantifications over $O(\log N)$ variables. Otherwise, there is only a constant number of operations.
\end{proof}

\subsection{Implicit Coloring of Interval Graphs}
\label{sec:algig_coloring}
Coloring refers to the task to color the nodes of a graph using the least number of colors, such that two adjacent nodes have different colors. Colorings of interval graphs are for example used in VLSI design (where it is called \textit{channel assignment}). In the case of interval graphs, there is an easy greedy coloring algorithm: Sort the endpoints of the intervals (\ie both left and right endpoints) in ascending order. At the beginning all colors are on a stack. Then color the intervals sequentially by traversing the sorted list and using the first color available on the stack when the current element is a left endpoint. As soon as we visit an right endpoint, we push the used color onto the top of the stack. This greedy algorithm is optimal and can be implemented to run in linear time by determining the order without sorting \cite{Olariu91}. The parallel algorithm in \cite{Zomaya1996} assigns weights to the endpoints and computes prefix sums to simulate the stack. In our implicit algorithm we can do the simulation in a more direct manner: We call two intervals $I_i = \left[ a_i,b_i\right]$ and $I_j = \left[ a_j,b_j\right]$ \emph{related} iff $b_i < a_j$ and $I_j$ is the first interval with the same color as $I_i$ in the greedy algorithm. The following easy observation helps us to compute this \enquote{related} relation implicitly.
\begin{observation}
\label{obs:num_rele}
The intervals $I_i$ and $I_j$ are two related iff the number of right endpoints $r$ with $b_i < r < a_j$ is equal to the number of left endpoints $l$ with $b_i < l < a_j$ and for all intervals $I_{j'}$ with $b_i < a_{j'} < a_j$ the number of right endpoints $r$ with $b_i < r < a_{j'}$ is not equal to the number of left endpoints $l$ with $b_i < l < a_{j'}$
\end{observation}
Now, in the case of unit intervals we want to show, how we can compute a function $RELATED(x,y)$, which is $1$ iff the interval $I_{\vert x \vert}$ and $I_{\vert y \vert}$ are related. The general case is discussed later in this section. As before, the intervals are labeled according to their left endpoints. Let $RE(x,y,l) = 1$ iff $\vert x \vert \leq \vert y \vert$ and the number of right endpoints between $b_{\vert x \vert}$ and $a_{\vert y \vert}$ is equal to $\vert l \vert$. Similarly, let $LE(x,y,l) = 1$ iff $\vert x \vert \leq \vert y \vert$ and the number of left endpoints between $b_{\vert x \vert}$ and $a_{\vert y \vert}$ is equal to $\vert l \vert$. Let $\chi_E(x,y)$ be the characteristic function of the edge set of a unit interval graph $G=(V,E)$ and $\chi_{E^c}(x,y)$ the characteristic function of the edge set of the complement graph, \ie $E^c = \lbrace (u,v) \mid u \neq v \text{ and } (u,v) \not\in E\rbrace$. Then we can compute $RE(x,y,l)$ and $LE(x,y,l)$ in the following way:
$$ \begin{array}{rcl}
H_1(x,y,z) & = & (\vert x \vert \leq \vert z \vert < \vert y \vert) \wedge \chi_{E^c}(z,y)\\
RE(x,y,l) & = & (\vert x \vert \leq \vert y \vert)\: \wedge\\
& & \left[ \exists z: H_1(x,y,z) \wedge (\vert z \vert - \vert x \vert = \vert l \vert) \wedge \overline{\exists z': H_1(x,y,z') \wedge (\vert z' \vert > \vert z \vert}) \right] \\
& &\\
H_2(x,y,z) & = & (\vert x \vert < \vert z \vert \leq \vert y \vert) \wedge \chi_{E^c}(x,z)\\
LE(x,y,l) & = & (\vert x \vert \leq \vert y \vert)\: \wedge\\
& & \left[\exists z: H_2(x,y,z) \wedge (\vert y \vert - \vert z \vert = \vert l \vert) \wedge \overline{\exists z': H_2(x,y,z') \wedge (\vert z' \vert < \vert z \vert}) \right].
\end{array}$$
The right endpoint of an interval $I_{\vert z \vert}$ is greater or equal than $b_{\vert x \vert}$ and less than $a_{\vert y \vert}$ iff $\vert x \vert \leq \vert z \vert < \vert y \vert$ and $I_{\vert z \vert}$ does not intersect $I_{\vert y \vert}$. Since we are dealing with unit interval graphs, if for some $z$ with $\vert x \vert \leq \vert z \vert < \vert y \vert$ the intervals $I_{\vert z \vert}$ and $I_{\vert y \vert}$ do not intersect, then it holds also for all $z'$ with $\vert x \vert \leq \vert z' \vert < \vert z \vert$. \Ie the maximal value of $\vert z \vert - \vert x \vert$ over all $z$ with the above property is equal to the number of right endpoints between $b_{\vert x \vert}$ and $a_{\vert y \vert}$ and, therefore, we compute the function $RE(x,y,l)$ correctly. Similar arguments show that $LE(x,y,l)$ is computed correctly, too. 

Together with Observation \ref{obs:num_rele}, we can compute the function $RELATED(x,y)$ as follows: 
$$ \begin{array}{rcl}
RELATED(x,y) & = & \exists l: RE(x,y,l) \wedge LE(x,y,l)\: \wedge\\
& & \overline{\exists z,l': (\vert z \vert < \vert y \vert) \wedge RE(x,z,l') \wedge LE(x,z,l')}
\end{array} $$
\noindent Now, we have to compute the sequence of related intervals, which is nothing more than the transitive closure of the related relation, which can be computed with $O(\log^2 N)$ functional operations. Finally, we have to assign a color to each interval, such that all intervals in a sequence of related intervals are getting the same color. In order to do this, we compute an order on the sequences of related intervals and assign the colors to the sequences according to that order. The order on the sequences is given by the order on the minimal interval number within the sequences. Putting all together, algorithm \ref{algo:unitcoloring} computes a coloring on a unit interval graph.

\begin{theorem}
Algorithm \ref{algo:unitcoloring} computes a coloring of a unit interval graph using the minimal number of colors and $O(\log^2 N)$ functional operations. 
\end{theorem}
\begin{proof}
That the output is a coloring with the minimal number of colors follows directly from correctness of the greedy algorithm. The number of functional operations is dominated by the $TransitiveClosure$ and $EnumerateOrder$ procedures. As we have seen in section \ref{sec:pre_basicalg}, both procedures need $O(\log^2 N)$ functional operations.
\end{proof}

\begin{algorithm}[H]
\caption{Implicit coloring algorithm for unit interval graphs}
\label{algo:unitcoloring}
\begin{algorithmic}[1]
\REQUIRE Unit interval graph $(\chi_E, \chi_V)$
\ENSURE Coloring $COLOR(x,l)$ with $COLOR(x,l) = 1$ iff $I_{\vert x \vert}$ has color $\vert l \vert$.
\\ \COMMENT{Complement graph}
\STATE $\chi_{E^c}(x,y) = \chi_V(x) \wedge \chi_V(y) \wedge \overline{\chi_E(x,y)} \wedge (\vert x \vert \neq \vert y \vert)$\\
\COMMENT {Auxiliary functions to compute the number of right/left endpoints}\\
\COMMENT{between two intervals}
\STATE $H_1(x,y,z) = (\vert x \vert \leq \vert z \vert < \vert y \vert) \wedge \chi_{E^c}(z,y) $\\
\STATE $H_2(x,y,z) = (\vert x \vert < \vert z \vert \leq \vert y \vert) \wedge \chi_{E^c}(x,z)$\\
\COMMENT {Number of right endpoints between $b_{\vert x \vert}$ and $a_{\vert y \vert}$}\\
\STATE $\begin{array}{rcl}
RE(x,y,l) & = & (\vert x \vert \leq \vert y \vert)\: \wedge\\
& & \left[ \exists z: H_1(x,y,z) \wedge (\vert z \vert - \vert x \vert = \vert l \vert) \wedge \overline{\exists z': H_1(x,y,z') \wedge (\vert z' \vert > \vert z \vert}) \right] 
\end{array}$\\
\COMMENT {Number of left endpoints between $b_{\vert x \vert}$ and $a_{\vert y \vert}$}\\
\STATE $\begin{array}{rcl}
LE(x,y,l) & = & (\vert x \vert \leq \vert y \vert)\: \wedge\\
& & \left[\exists z: H_2(x,y,z) \wedge (\vert y \vert - \vert z \vert = \vert l \vert) \wedge \overline{\exists z': H_2(x,y,z') \wedge (\vert z' \vert < \vert z \vert}) \right]
\end{array}$\\
\COMMENT {Compute related intervals}
\STATE $\begin{array}[t]{lcl}
RELATED(x,y) & = & \exists l: RE(x,y,l) \wedge LE(x,y,l)\: \wedge\\
& & \overline{\exists z,l': (\vert z \vert < \vert y \vert) \wedge RE(x,z,l') \wedge LE(x,z,l')}
\end{array}$\\
\COMMENT{Compute set of intervals with the same color}
\STATE $SAMECOLOR(x,y) = TransitiveClosure(RELATED(x,y) \vee (\vert x \vert = \vert y \vert))$\\
\COMMENT{Order these sets}
\STATE $FIRST(x) = \overline{\exists x': SAMECOLOR(x',x) \wedge (\vert x' \vert < \vert x \vert)}$\\
\STATE $\begin{array}{lcl}
COLORORDER(x,y) & = & \exists x',y': SAMECOLOR(x',x) \wedge FIRST(x')\: \wedge\\
& & SAMECOLOR(y',y) \wedge FIRST(y') \wedge (\vert x' \vert < \vert y' \vert)
\end{array}$\\
\COMMENT{Assign the colors}
\STATE  $COLOR(x,l) = EnumerateOrder(COLORORDER(x,y))$
\RETURN $COLOR(x,l)$
\end{algorithmic}
\end{algorithm}
\noindent The only difference between the unit interval and the general case is the computation of the functions $LE$ and $RE$ (this is the only place where we need the unity property). What we actually need, is an order on the sequence of right endpoints to compute $RE$ and an order on the left endpoints of the intervals to compute $LE$ (and in the case of unit intervals both orders are the same). Assuming that we label the intervals according to their left endpoints, we only need to compute the order on the right endpoints. Let $EO(x,y)$ be this order, \ie $EO(x,y) = 1$ iff $b_{\vert x \vert} \leq b_{\vert y \vert}$. Remember the adjacency matrix of an interval graph from section \ref{sec:sizeintervalgraph} and assume that the left points of the intervals are the integers $0,\ldots,N-1$. We know that the interval with left endpoint $i \in \lbrace 0,\ldots, N-1 \rbrace$ has a maximal value $j$ such that $I_i$ and $I_{k}$ intersect for all $i \leq k \leq j$. Therefore, the right endpoint of $I_i$ has to be in $[j,j+1)$. Let $j$ and $j'$ be the maximal values such that $I_i$ intersects all $I_k$ with $i \leq k \leq j$ and $I_{i'}$ intersects all $I_k$ with $i' \leq k \leq j'$, respectively. If $j < j'$ ($j > j'$), then $b_i < b_{i'}$ ($b_i > b_{i'}$). If $j = j'$, then we can break ties arbitrary (\eg $b_i \leq b_{i'}$ iff $i \leq i'$). 
Now, we can compute $EO(x,y)$ as follows:
$$ \begin{array}{rcl}
H(x,y,x',y') & = & (\vert x \vert \leq \vert x' \vert) \wedge (\vert y \vert \leq \vert y' \vert) \wedge \chi_{E}(x,x') \wedge \chi_E(y,y')\\
EO(x,y) & = &  \exists x',y': H(x,y,x',y') \wedge (\vert x' \vert < \vert y' \vert \vee (\vert x' \vert = \vert y' \vert \wedge \vert x \vert < \vert y \vert))\: \wedge\\
& & \overline{\exists x'',y'': H(x,y,x'',y'') \wedge ((\vert x'' \vert > \vert x' \vert) \vee (\vert y'' \vert > \vert y' \vert))}
\end{array}$$
Notice that this order on the right endpoints does not have to be the same order on the original right endpoints. But, as we have shown, there is an interval representation of the graph, such that the left and right endpoints are ordered according to labels of the nodes and $EO(x,y)$, respectively. Finally, we have to compute $EEO(x,l) = EnumerateOrder(EO(x,y))$ and, with it, we get $RE(x,y,l)$ for general interval graphs:
$$ \begin{array}{rcl}
H_1(x,y,z) & = & \overline{EO(z,x)} \wedge EO(z,y) \wedge \chi_{E^c}(z,y)\\
RE(x,y,l) & = & (\vert x \vert < \vert y \vert) \wedge \left[ \exists z,l_1,l_2: H_1(x,y,z) \wedge EEO(x,l_1)\: \wedge \right.\\
& & \left. EEO(z,l_2) \wedge (\vert l_2 \vert - \vert l_1 \vert = \vert l \vert) \wedge \overline{\exists z': H_1(x,y,z') \wedge EO(z,z')} \right]
\end{array}$$
Since all additional operations are dominated by the $EnumarateOrder$ procedure, we get the same result as for unit intervals.
\begin{theorem}
Algorithm \ref{algo:unitcoloring} with the modified computation of $RE(x,y,l)$ outputs a coloring of an interval graph using the minimal number of colors and $O(\log^2 N)$ functional operations. 
\end{theorem}
\section{Experimental Evaluation}
We evaluated the implicit maximum matching algorithm on unit interval graphs and the implicit coloring algorithm on unit and general interval graphs. Unfortunately, the implicit coloring algorithm performed poorly even on instances of size around $2000$. Therefore, we only show the results for the maximum matching algorithm but want to begin with a discussion of this performance difference: At a first glance, this might not be surprising due to the more complex coloring algorithm but having a closer look we see that the implicit matching algorithm is optimized for the implicit setting while the implicit coloring algorithm uses some nice ideas to simulate the sequential algorithm. Hence, these results do not rule out the possibility of an efficient implicit coloring algorithm but suggest that there have to be new ideas to benefit more from the strengths of implicit algorithms.

Unit interval graphs can be represented as balanced nonnegative strings over $\lbrace \text{`['} , \text{ ']'} \rbrace$ (see, \cite{SaitohYKU10}) and such strings are created randomly using the algorithm in \cite{ArnoldS80}. We generated $35$ random graphs of size $2^i$ for $i=10,\ldots,23$. The nodes of the graphs are encoded as in Section \ref{sec:sizeintervalgraph}. We compare the OBDD-based algorithm to the algorithm which gets the interval representation as an input, sort the intervals according to their starting point and compute a maximum matching by scanning this sorted sequence with the same idea used in the implicit algorithm.
\subsubsection*{Experimental Setup}
We implemented the implicit algorithm with the BDD framework CUDD 2.5.0\footnote{\url{http://vlsi.colorado.edu/~fabio/CUDD/}} by F. Somenzi. The algorithms are implemented in C++ and were compiled with Visual Studio 2012 in the default release configuration. All source files, scripts and random seeds are publicly available\footnote{\url{http://ls2-www.cs.uni-dortmund.de/~gille/}}. The experiments were performed on a computer with a 2.6 GHz Intel Core i5 processor and 4 GB main memory running Windows 7. The runtime is measured by used processor time in seconds and the space usage of the implicit algorithm is given by the maximum SBDD size which came up during the computation, where a SBDD is a collection of OBDDs which can share nodes. Due to the small variance of these values, we only show the mean in the diagrams.
\begin{figure}[t]
\centering
\includegraphics[width=0.48\textwidth]{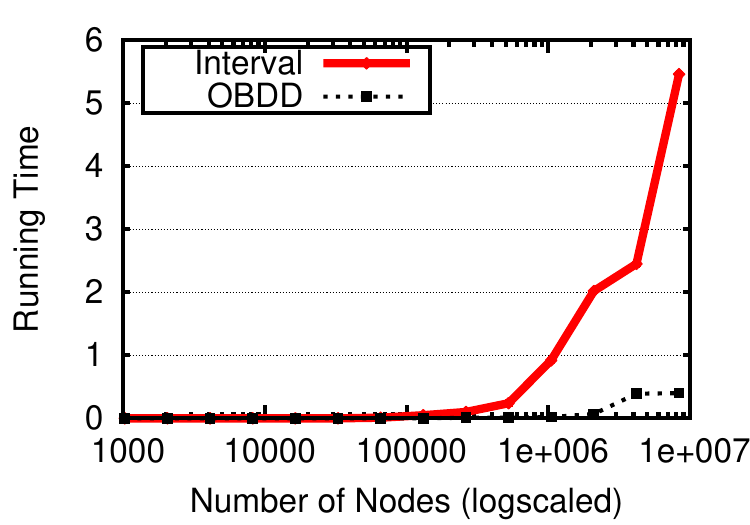}
\includegraphics[width=0.48\textwidth]{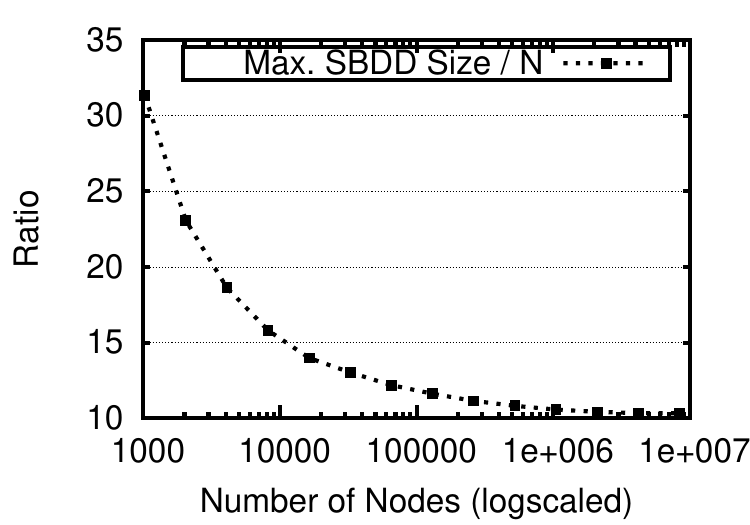}
\caption{Runtime and memory of the matching algorithms on random unit interval graphs. Memory plot shows the ratio of $S \log S$ (space usage of the OBDD-based algorithm) and number of nodes.}\label{fig:diag}
\end{figure}
\subsubsection*{Results}
The implicit matching algorithm outperforms the explicit matching algorithm on unit interval graphs (see Fig. \ref{fig:diag}). Even on graphs with more than 8 million nodes the implicit algorithm computes a maximum matching within 1 seconds. Storing a SBDD of size $S$ needs $O(S \log S)$ bits. The memory diagram shows that the asymptotic space usage of the implicit algorithm on these instances is close to $O(N)$. Recall that the unit interval representation needs $\Theta(N \log N)$ space since $\log N$ bits are needed to represent the starting points. \Ie the implicit algorithm needs less space and can compute a maximum matching on larger instances than the explicit one. An interesting consequence of these results is that the submodules of our maximum matching algorithm, namely computing the connected components, a Hamiltonian path in every connected component and a maximum matching on these paths, are also very fast and space efficient which is surprising, since especially the computation of the transitive closure is often a bottleneck in implicit algorithms.
\section{Conclusion and Open Questions}
In this paper, we presented a method to show upper bounds of the size of OBDDs representing a graph by using the adjacency matrix. Using this method, we could improve known results on the OBDD size of interval graphs and we think that it is possible to show similar results for other graph classes with a well structured adjacency matrix, \eg convex graphs where the nodes can be ordered such that the neighborhood of every node consists of nodes which are consecutive in this order. The gap between the upper and lower bound (using another labeling or variable order) of the OBDD size of interval graphs is $O(\log N)$. It is an interesting open question whether there is another labeling and/or variable order such that the OBDD size is $O(N)$ or the general lower bound can be increased to $\Omega(N \log N)$, which we believe is more likely, since the observation that the columns of the $\pi$-ordered adjacency matrix are independent also holds for an arbitrary labeling.

Even for a fixed variable order, the complexity of computing a node labeling for a given graph, such that the representing OBDD has minimal size, is unknown. The $\pi$-ordered adjacency matrix seems to help to prove upper/lower bounds on the OBDD size for a fixed labeling. Using this matrix to bound the size of OBDDs for every labeling could be object of further research.

The parallel maximum matching algorithm on general interval graphs \cite{ChungPC97} is more complex than the parallel maximum matching algorithm on unit interval graphs and the algorithm does not seem to be directly applicable to develop an implicit algorithm as for unit intervals.

The investigation of implicit algorithms on special graph classes seems quite promising and it would be interesting if the good performance can also be achieved for other large graph classes.

\subsubsection*{Acknowledgements} I thank Beate Bollig, Melanie Schmidt and Chris Schwiegelshohn for the valuable discussions and, together with the anonymous referees, for their comments on the presentation of the paper.

\bibliographystyle{acm} 
\bibliography{literatur}

\end{document}